\documentclass[twoside,leqno,twocolumn]{article}  
\usepackage{ltexpprt} 
\usepackage{nicefrac}
\usepackage{xspace}
\usepackage{amsmath,amstext,amssymb,amsfonts}
\usepackage{tabularx}
\usepackage{verbatim}
\usepackage{graphicx}

\newcommand{\nfrac}{\nicefrac}

\usepackage[varg]{txfonts}

\renewcommand{\mathbb}{\varmathbb}

   \usepackage{fullpage}

\renewcommand{\leq}{\leqslant}

\renewcommand{\geq}{\geqslant}

\usepackage{bm}

\usepackage{xspace}

 \usepackage[top=2.5cm, bottom=2.5cm, left=2cm, right=2cm]{geometry}

\newcommand{\cE}{\mathcal E}

\newcommand{\cI}{\mathcal I}

\newcommand{\cS}{\mathcal S}

\newcommand{\abs}[1]{\lvert#1\rvert}

\newcommand{\floor}[1]{\lfloor #1 \rfloor}

\newcommand{\norm}[1]{\lVert#1\rVert}
\newcommand{\Norm}[1]{\left\lVert#1\right\rVert}

\newcommand{\supp}{\mathrm{supp}}
\newcommand{\poly}{\mathrm{poly}}
\newcommand{\vol}{\mathrm{vol}}

\newcommand{\defeq}{\stackrel{\textup{def}}{=}}

\newcommand{\normt}[1]{\norm{#1}_{\scriptstyle 2}}

\newcommand{\R}{\mathbb R}

\newcommand{\sdp}{{\sf sdp}}
\newcommand{\sgn}{{\rm sgn}}

\newcommand{\card}{\abs}

\newcommand{\Esymb}{\mathbb{E}}
\newcommand{\Psymb}{\mathbb{P}}

\DeclareMathOperator*{\E}{\Esymb}

\DeclareMathOperator*{\ProbOp}{\Psymb}
\renewcommand{\Pr}{\ProbOp}
\newcommand{\Tr}{{\rm Tr}}

\newcommand{\e}{\epsilon}

\let\e\varepsilon

\newcommand\bdot\bullet

\newcommand{\BS}{{\sc Balanced Separator}\xspace}

\newcommand{\SDP}{{\sf SDP}\xspace}

\newcommand{\Vol}[1]{\text{vol}(#1)}

\newcommand{\Deg}{D^{\nfrac{1}{2}}}
\newcommand{\Degin}{D^{-\nfrac{1}{2}}}

\numberwithin{equation}{section}

\newcommand{\psdp}{{\ensuremath{{\mathsf{psdp}}}\xspace}}
\newcommand{\dsdp}{{\ensuremath{\mathsf{dsdp}}\xspace}}

\newcommand{\alg}{{\sc BalCut}\xspace}
\newcommand{\avg}{{\ensuremath{\mathsf{avg}}\xspace}}

\begin{document}

\title{
  \textsc{%
Towards an SDP-based Approach to Spectral Methods}\\ {\Large A Nearly-Linear-Time Algorithm for Graph Partitioning and Decomposition}
}

\author{Lorenzo Orecchia  \thanks{UC Berkeley. Supported by NSF grants CCF-0830797 and CCF-0635401. This work was initiated while this author was visiting Microsoft Research India, Bangalore.}\\
\and 
Nisheeth K. Vishnoi \thanks{Microsoft Research India, Bangalore.}}
\date{}

\maketitle	

\begin{abstract}
In this paper, we consider the following graph partitioning problem: The input is an undirected graph $G=(V,E),$ a balance parameter $b \in (0,1/2]$  and a target conductance value $\gamma \in (0,1).$ 
The  output is a cut which, if non-empty, is of conductance at most $O(f),$ for some function $f(G, \gamma),$ and which  is either balanced or well correlated with all cuts of conductance at most $\gamma.$ 
In a seminal paper, Spielman and Teng \cite{ST1} gave an $\tilde{O}(|E|/\gamma^{2})$-time algorithm for $f= \sqrt{\gamma \log^{3}|V|}$ and used it to decompose graphs into a collection of near-expanders \cite{ST2}.

We present a new spectral algorithm for this problem which runs in time $\tilde{O}(|E|/\gamma)$ for $f=\sqrt{\gamma}.$ Our result yields the first nearly-linear time algorithm for the classic \BS problem that achieves the asymptotically optimal approximation guarantee for spectral methods.

Our method has the advantage of being conceptually simple and relies on a primal-dual semidefinite-programming (\SDP) approach. 
We first consider a natural SDP relaxation for the \BS problem.  While it is easy to obtain from this SDP a certificate of the fact that the graph has no balanced cut of conductance less than $\gamma,$ somewhat  surprisingly,  we can obtain a certificate for the stronger correlation condition. 
This is achieved via a novel separation oracle for our SDP and by appealing to Arora and Kale's \cite{AK}  framework to bound the running time.
Our result contains technical ingredients that may be  of independent interest.
\end{abstract}

\section{Introduction}

\subsection{Graph Partitioning.}
Given a graph $G=(V,E),$  the conductance of
a cut $(S,\bar{S})$ is  $\phi(S) \defeq \nfrac{|E(S,\bar{S})|}{\min
\{\vol(S),\vol (\overline{S})\}},$ where $\vol(S)$ is the sum of the degrees of the vertices in the set $S$. 
A cut $(S, \bar{S})$ is $b$-balanced if $\min \{\vol(S), \vol(\bar{S})\} \geq b \cdot \vol{V}.$ 
A graph partitioning problem of widespread interest is the  \BS  problem: given $G=(V,E),$ a constant\footnote{
We will use $O_b(\cdot)$ and $\Omega_b(\cdot)$ in our asymptotic notation when we want to emphasize the dependence of the hidden coefficent on $b.$}
balance parameter $b \in (0,\nfrac{1}{2}],$ and a conductance value $\gamma \in (0,1),$ does $G$ have a $b$-balanced cut $S$ such that $\phi(S) \leq \gamma$?

\BS is an intensely studied problem in both theory and practice.  
It has far-reaching connections to spectral graph theory, the study of random walks and metric embeddings. Besides being a theoretically rich problem, \BS is of great practical importance, as it plays a central role in the design of recursive algorithms, image segmentation and clustering. 

Since {\BS}  is an  NP-hard problem~\cite{GareyJ79}, we seek approximation algorithms that either output a cut of conductance at-most $f(\gamma,\log |V|)$ and balance $
\Omega_b(1)$ or a certificate that $G$ has no $b$-balanced cut of conductance at most $\gamma.$ 
In their seminal series of papers~\cite{ST1, ST2, ST3}, Spielman and Teng use an approximation algorithm for {\BS} as a fundamental primitive to decompose the instance graph into a collection of near-expanders. This decomposition is then used to construct spectral sparsifiers and solve systems of linear equations in nearly linear time. Their algorithm has two crucial features: first, it runs in nearly linear time; second, in the case that no balanced cut exists in the graph, it outputs a certificate of a special form. This certificate consists of an unbalanced cut of small conductance which is well-correlated with all low-conductance cuts in the graph. We prove in Section \ref{app:cut} in the Appendix that such a cut is indeed a negative certificate for the \BS problem.
Formally, they prove the following:
\begin{theorem}\cite{ST1} \label{thm:st}
Given a graph  $G$, a balance parameter $b \in (0,\nfrac{1}{2}], \; b = \Omega(1)$ and a conductance value $\gamma \in (0,1),$ {\sc Partition}$(G,b, \gamma)$ runs in time $\tau$ and outputs a cut $S \subseteq V$ such that $\vol(S) \leq \nfrac{7}{8}\cdot \vol(G),$  $\phi(S) \leq f_{1}$ or $S = \emptyset,$  and with high probability, either
\begin{enumerate}

	\item $S$ is $\Omega_b(1)$-balanced, or 
	\item for all $C \subset V$ such that $\vol(C) \leq \nfrac{1}{2} \cdot  \vol(G)$ and  $\phi(C) \leq O(\gamma),$ 
$
\frac{\vol(S \cap C)}{\vol(C)} \geq \nfrac{1}{2}.
$
\end{enumerate}
\end{theorem}

\noindent
Originally, Spielman and Teng showed Theorem \ref{thm:st} with $f_1= O\left( \sqrt{\gamma \log^3 n}\right)$ and $\tau = \tilde{O}\left(\nfrac{m}{\gamma^2}\right).$ This was subsequently improved by Andersen, Chung and Lang \cite{ACL} and then by Andersen and Peres \cite{AP} to the current best of $f_1=O\left(\sqrt{\gamma \log n}\right)$ and $\tau = \tilde{O}(\nfrac{m}{\sqrt{\gamma}}).$ All these results made use of bounds on the convergence of random walk processes on the instance graph, such as the Lovasz-Simonovits bounds \cite{LS}. These bounds yield the $\log n$ factor in the approximation guarantee, which appears hard to remove while following this approach.

\subsection{Our Contribution}

In this paper,  we use a semidefinite programming approach to design a new spectral algorithm, called {\sc BalCut}, that improves on the result of Theorem \ref{thm:st}. The following is our main result.

\begin{theorem}[Main Theorem] \label{thm:main}
Given a graph  $G=(V,E)$, a balance parameter $b \in (0,\nfrac{1}{2}], \; b = \Omega(1),$ and a conductance value $\gamma \in (0,1),$ {\sc BalCut}$(G,b, \gamma)$ runs in time $\tilde{O}\left(\nfrac{m}{\gamma}\right)$ and outputs a cut $S \subset V$  such that $\vol(S) \leq \nfrac{7}{8}\cdot \vol(G),$  if $S \neq \emptyset$ then  $\phi(S) \leq O_b\left(\sqrt{\gamma}\right),$ and with high probability, either
\begin{enumerate}
\item $S$ is $\Omega_b(1)$-balanced, or 
\item for all $C \subset V$ such that $\vol(C) \leq \nfrac{1}{2} \cdot \vol(G) $ and  $\phi(C) \leq O(\gamma),$ 
$
\frac{\vol(S \cap C)}{\vol(C)} \geq \nfrac{1}{2}.
$
\end{enumerate} 
\end{theorem}

\noindent
Note that our result improves the parameters  of previous algorithms by eliminating the $\log n$ factor in the quality of the cut output, making the approximation comparable to the best that can be hoped for using spectral methods~\cite{GM}. Our result is also conceptually simple: we use the primal-dual framework of Arora and Kale \cite{AK} to solve {\SDP}s combinatorially, and we give a new separation oracle that yields  Theorem \ref{thm:main}.
Finally, our result implies an approximation algorithm for {\BS}, as the guarantee of Theorem \ref{thm:main} on the cut $S$ output by \alg also implies a lower bound on the conductance of balanced cuts of $G.$ The proof can be found in Section \ref{app:cut} in the Appendix.
\begin{corollary}\label{cor:cut}
Given an instance graph $G,$ a balance parameter $b \in (0, \nfrac{1}{2}]$ and a target conductance $\gamma \in (0,1],$ \alg$(G,b,\gamma)$ either outputs an $\Omega_b(1)$-balanced cut of conductance at most $O_b(\sqrt{\gamma})$ or a certificate that all $\Omega_b(1)$-balanced cuts have conductance at least $\Omega_b(\gamma).$ The running time of the algorithm is $\tilde{O}(m/{\gamma}).$
\end{corollary}
This is the first nearly-linear-time spectral algorithm for \BS that achieves the asymptotically optimal approximation guarantee for spectral methods.

\subsection{Graph Decomposition.} The main application of Theorem \ref{thm:st} is the construction of a particular kind of graph decomposition. In this decomposition, we wish to partition the vertex set of the instance graph $V$ into components $V_1, \ldots, V_i, \ldots, V_k$ such that the graph induced by $G$ on each $V_i$ has conductance as large as possible, while at most a constant fraction of the edges have endpoints in different components. These decompositions are a useful algorithmic tool in several areas \cite{Trevisan05, KM, ST2}.

Kannan, Vempala and Vetta \cite{Kannan} construct such decompositions achieving a conductance value of $\Omega(\nfrac{1}{\log^2 n}).$ However, their algorithm runs in time $\tilde{O}(m^2)$ on some instances. 

Spielman and Teng~\cite{ST2} relax this notion of decomposition by only requiring that each $V_i$ be contained in a superset $W_i$ in $G$, where $W_i$ has large induced conductance in $G.$ In the same work, they show that this relaxed notion of decomposition suffices for the purposes of sparsification by random sampling. The advantage of this relaxation is that it is now possible to compute this decomposition in nearly-linear time by recursively applying the algorithm of Theorem \ref{thm:st}.
\begin{theorem} \label{thm:dec} \cite{ST2} 
Assume the existence of an algorithm achieving parameters $\tau$ and $f_1$ in Theorem \ref{thm:st}. Given $\gamma \in (0,1),$ in time $\tilde{O}(\tau),$ it is possible to construct a decompositions of the instance graph $G$ into components $V_1, \ldots, V_k$ such that:
\begin{enumerate}
\item for each $V_i,$ there exists $W_i \supseteq V_i$ such that the conductance of the graph induced by $G$ on $W_i$ is  $\Omega(\nfrac{\gamma}{\log n}).$
\item the fraction of edges with endpoints in different components is $O(f_1 \log n).$
\end{enumerate}
\end{theorem}
Using Theorem~\ref{thm:dec}, Spielman and Teng showed the existence of a decomposition achieving conductance $\Omega(\nfrac{1}{\log^6 n}).$
Our improved results in Theorem \ref{thm:main} imply that we can obtain decompositions of the same kind with conductance  bound $\Omega(\nfrac{1}{\log^3 n}).$ Our improvement also implies speed-ups in the sparsification procedure described by Spielman and Teng~\cite{ST2}. However, this result has since been superceded by work of Koutis, Miller and Peng~\cite{KMP} that gives a very fast linear equation solver that can be used to compute sampling probabilities for each edge, yielding a spectral sparsifier with high probability~\cite{SS}.

Our work leaves open the important question posed by Spielman~\cite{SpielmanICM} of whether stronger decompositions, of the kind proposed by Kannan, Vempala and Vetta~\cite{Kannan}, can be produced in nearly-linear time.

\subsection{Overview of Techniques}
\paragraph{Spectral Approach.}
The simplest algorithm for \BS, also used by Kannan et al.~\cite{Kannan}, is the recursive spectral algorithm. 
This algorithm finds the minimum-conductance sweep cut of the second eigenvector of $G$,  removes the cut and all adjacent edges from $G,$ and reiterates on the remaining graph. The algorithm stops when the union of the cuts removed becomes $\nfrac{b}{2}$-balanced or when the residual graph is found to have spectral gap at least $\gamma,$ certifying that no more progress can be made. As every cut may only remove $O(1)$ volume  and the eigenvector computation takes $\Omega(m)$ time, this algorithm may have quadratic running time. It can be shown using Cheeger's Inequality \cite{FAN} that the cut this procedure outputs is of conductance at most $O(\sqrt{\gamma}).$ 
\paragraph{Spielman-Teng Approach.}
The algorithm of Spielman and Teng which proves Theorem~\ref{thm:st} is also spectral in nature and  uses, as the main subroutine, {\it local random walks} that run in time proportional to the volume of the output cut to find sparse cuts around vertices of the graphs. These local methods are based on  non-trivial random walks on the input graph and aggregation of the information obtained from these walks, all performed while maintaining  nearly-linear running time.
\paragraph{Our Approach.} 
We depart from the random-walk paradigm and first consider  a natural {{\SDP}} relaxation for the {\BS} problem, which \alg solves approximately using a primal-dual method.
Intuitively, \alg manages to maintain the approximation guarantee of the recursive spectral algorithm while running in nearly-linear time by considering a distribution over eigenvectors,  represented as a vector embedding of the vertices, rather than a single eigenvector, at each iteration. The sweep cut over the eigenvector is replaced by a sweep cut over the radius of the vectors in the embedding (see Figure~\ref{fig:hedging}).
Moreover, at any iteration, rather than removing the unbalanced cut found, \alg penalizes it by modifying the graph so that it is unlikely but still possible for it to turn up in future iterations. Hence, in both its cut-finding and cut-eliminating procedures, \alg tends to ``hedge its bets" more than the greedy recursive spectral method. This hedging, which ultimately allows \alg to achieve its faster running time, is implicit in the primal-dual framework of Arora and Kale~\cite{AK}.

\begin{figure}[!h]
\begin{center}
\includegraphics[ clip=true, scale=.45]{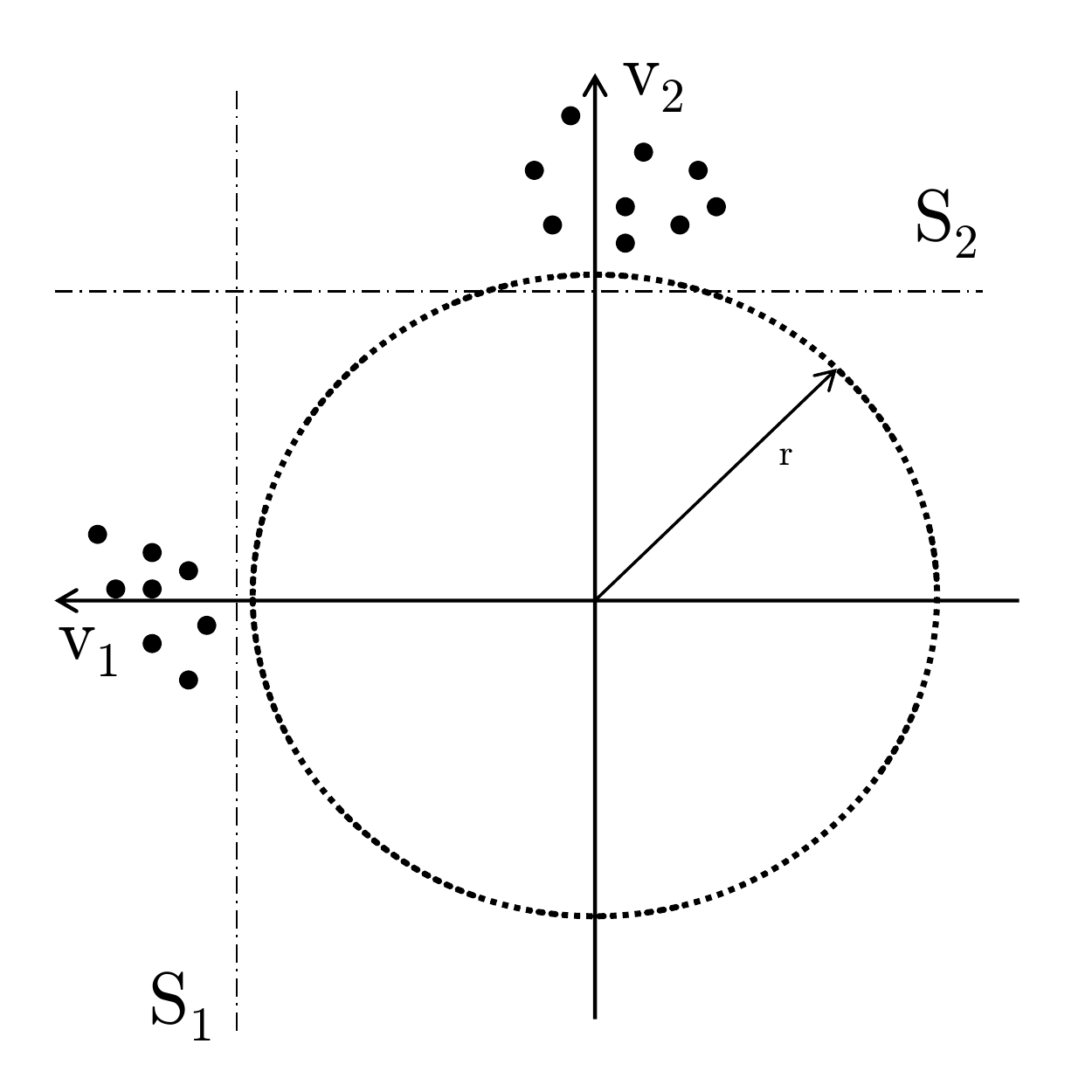}
\caption{Schematic representation of the speed-up introduced by \alg when the instance graph contains many unbalanced cuts of low conductance. Let $v_1$ and $v_2$ be the two slowest-mixing eigenvectors of $G.$ Assume that their minimum-conductance sweep cuts $S_1$ and $S_2$ are unbalanced cuts of conductance less than $\gamma.$ If we use the recursive algorithm of Kannan et al.~\cite{Kannan}, two iterations could be required to remove $S_1$ and $S_2$. However, \alg considers a multidimensional embedding containing contributions from multiple eigenvectors and performs a radial sweep cut. This allows $S_1$ and $S_2$ to be removed in a single iteration.}
\label{fig:hedging}
\end{center}
\vspace{-8mm}
\end{figure}

The {{\SDP}} relaxation appears in Figure \ref{fig:sdp-intro}.
We denote by $\mu:V\mapsto \mathbb{R}_{\geq 0}$ the distribution defined as $\mu_{i}\defeq \nfrac{d_{i}}{\vol(G)},$ and by $d_{i}$ the degree of the $i$-th vertex. Also, $v_{avg}\defeq \sum_{i}\mu_{i}v_{i}.$
Even though our algorithm uses the {{\SDP}}, at the core, it is spectral in nature, as it relies on the matrix-vector multiplication primitive. Hence, if one delves deeper, a random walk interpretation can be derived for our algorithm.
\begin{figure}[htb]

\begin{alignat*}{4}
 \psdp(G, b, \gamma): & & \nfrac{1}{4} \cdot \E_{\{i,j\} \in E}\;\norm{v_i-v_j}_2^2 \; \leq& \; \gamma   \\
 & & \E_{j \sim \mu}\; \norm{v_j- v_\avg}_2^2 \;=& \;1 \\
    &\forall {i\in V} \;   &  \norm{v_i- v_\avg}_2^2 \;\leq& \; \frac{(1-b)}{b}
\end{alignat*}
\label{fig:sdp-intro}
\caption{{\SDP} for $b$-\BS}
\end{figure}


\paragraph{The Primal-Dual Framework.}
For our {\SDP}, the method of Arora and Kale can be understood as a game between two players: an embedding player and an  oracle player. The embedding player, in every round of this game, gives a candidate vector embedding of the vertices of the instance graph to the oracle player. We show that, if we are lucky and  the embedding is feasible for the {\SDP} and, in addition, also has the property that for a large set $S,$ for every $i \in S,$ $\|v_{i}-v_{avg}\|_{2} \leq O(\nfrac{(1-b)}{b})$  (we call such an embedding roundable), then a projection of the vectors along a random direction followed by a  sweep cut  gives an $\Omega_{b}(1)$-balanced cut of conductance at most $O(\sqrt{\gamma}).$ The  difficult case is when the embedding given to the oracle player is not roundable. In this case, the oracle outputs a candidate dual solution along with a cut. The oracle obtains this  cut  by performing a radial sweep cut of the vectors given by the embedding player.  If at any point in this game the union of cuts output by the oracle becomes balanced, we output this union and stop. We show that such a cut is of conductance at most $O_b(\sqrt{\gamma}).$  If this union of cuts is not balanced, then the embedding player uses the dual solution output by the oracle to update the embedding.
Finally, the matrix-exponential update rule ensures that this game cannot keep on going  for more that $O(\nfrac{\log n}/{\gamma})$ rounds. Hence, if a balanced cut is not found after this many rounds, we certify that the graph does not contain any $b$-balanced cut of conductance less than $\gamma.$  
To achieve a nearly-linear running time, we  maintain only  a $\log n$-dimensional sketch of the embedding.  The guarantee on the running time then follows by noticing that,  in each iteration,  the most expensive computational step for each player is a logarithmic number of matrix-vector multiplications, which takes at most $\tilde{O}(m)$ time. 

The reason why our approach yields the desired correlation condition in Theorem \ref{thm:main} is that, if no balanced cut is found, every unbalanced cut of conductance lower than $\gamma$ will, at some iteration, have a lot of its vertices mapped to vectors of large radius. At that iteration, the cut output by the oracle player will have a large correlation with the target cut, which implies that the union of cuts output by the oracle player will also display such large correlation. This intuition is formalized in the proof of Theorem \ref{thm:main}.



\paragraph{Our Contribution.}
The implementation of the oracle player, specifically dealing with the case when the embedding is not roundable,   is the main technical novelty of the paper. Studying the problem in the SDP-framework is the main conceptual novelty. 
The main advantage of using {\SDP}s to design a spectral algorithm seems to be that {\SDP} solutions provide a  simple representation for possibly complex random-walk objects. 
Furthermore, the benefits of using a carefully designed {\SDP} formulation can often  be reaped with little or no burden on the running time of the algorithm, thanks to the primal-dual framework of Arora and Kale~\cite{AK}.

\subsection{Rest of the Paper}
In Section~\ref{sec:notation}, we set the notation for the paper. In Section~\ref{sec:sdp},  we present our {\SDP} and its dual, and also define the notion of a roundable embedding. In Section~\ref{sec:pd},  we present the algorithm {\sc BalCut} and the separation oracle {\sc Oracle}, and reduce the task of  proving  Theorem \ref{thm:main} to proving statements about the {\sc Oracle}. Section~\ref{sec:mp} contains the proof of the main theorem about the {\sc Oracle} used in Section~\ref{sec:pd}.  For clarity of presentation, several proofs are omitted from the above sections and appear in the appendix.

\section{Algorithm Statement and Main Theorems}

\subsection{Notation} \label{sec:notation}

\paragraph{Instance graph and edge volume.} We denote by $G=(V,E)$ the unweighted instance graph, where $\card{V} = n$ and $\card{E}=m.$ We let $d \in \R^V,$ be the degree vector of $G,$ i.e. $d_i$ is the degree of vertex $i.$ We mostly work with the edge measure $\mu$ over $V,$ defined as $\mu_i \defeq \mu(i) \defeq \nfrac{d_i}{2m}.$ For a subset $ S \subseteq V,$ we also define $\mu_S$ as the edge measure over $S$, i.e. $\mu_S(i) \defeq \nfrac{\mu(i)}{\mu(S)}.$

\paragraph{Special graphs}
For a subset $S \subseteq V,$ we denote by $K_S$ the complete graph over $S$ such that edge $\{i,j\}$ has weight $\mu_i \mu_j$ for $i,j \in S$ and $0$ otherwise. $K_V$ is the complete graph with weight $\mu_i \mu_j$ between every pair $i,j \in V.$

\paragraph{Graph matrices.} For an undirected graph $H=(V, E_H)$, let
$A(H)$ denote the adjacency matrix of $H$ and $D(H)$ the diagonal matrix of
degrees of $H$.
The (combinatorial) Laplacian of $H$ is defined as $L(H) \defeq D(H) - A(H)$.
Note that for all $x \in \R^{V}$, $x^T L(H) x = \sum_{\{i,j\} \in E_H} (x_i - x_j)^2$.
By $D$ and $L$, we denote $ D(G)$ and $L(G)$ respectively.

\paragraph*{Vector and matrix notation.}

For a symmetric matrix $M,$ we will use $M \succeq 0$ to denote that it is positive semi-definite and $M \succ 0$ to denote that it is positive definite.
The expression $A \succeq B$ is equivalent to $A - B \succeq 0$. For two matrices $A,B$ of equal dimensions, denote $A \bullet B \defeq \Tr(A^T B) = \sum_{ij} A_{ij}\cdot B_{ij}.$ 
%
For a matrix $A$, we indicate by $t_A$ the time necessary to compute the matrix-vector multiplications $Au$ for any vector $u$. 
%

\paragraph{Embedding notation.} We will deal with vector embeddings of $G$, where each vertex $i \in V$ is mapped to a vector $v_i \in \R^d.$ For such an embedding $\{v_i\}_{i \in V},$ we denote by $v_\avg$ the mean vector, i.e. $v_\avg \defeq \sum_{i \in V} \mu_i v_i.$
Given a vector embedding of $\{v_i \in \R^d\}_{i \in V},$ recall that $X \succeq 0,$ is the Gram matrix of the embedding if $X_{ij} = v_i^T v_j.$ For any $X \in \R^{V \times V}, X\succeq 0,$ we call $\{v_i\}_{i \in V}$ the {\it embedding corresponding to $X$} if $X$ is the Gram matrix of $\{v_i\}_{i \in V}.$ 
For $i \in V,$ we denote by $R_i$ the matrix such that $R_i \bullet X = \norm{v_i - v_\avg}_2^2.$

\paragraph{Basic facts.} We will alternatively use vector and matrix notation to reason about the graph embeddings. 
The following are some simple conversions between vectors and matrix forms and some basic geometric facts which follow immediately from definitions.
\begin{fact}\label{fct:mean}
$ \E_{i \sim \mu}\norm{v_i - v_\avg}_2^2 = \nfrac{1}{2} \cdot \E_{\{i,j\} \sim \mu \times \mu}\norm{v_i - v_j}_2^2 = L(K_V) \bullet X.$
\end{fact}
%
%
%
%
\begin{fact} \label{fct:star} For a subset $S \subseteq V,$ 
$\sum_{i \in S} \mu_i R_i \succeq \mu(S) L(K_V) - L(K_S).$
\end{fact}
\begin{fact} \label{fct:subset}
For a subset $S \subseteq V,$ $ \E_{\{i,j\} \sim \mu_S \times \mu_S}\norm{v_i - v_j}_2^2 = 2 \cdot \nfrac{1}{\mu(S)^2} \cdot L(K_S) \bullet X.$
\end{fact}

\paragraph{Modified matrix exponential update.}
Let $\cS_D$ e the subspace of $\R^V$ orthogonal to $\hat{v} \defeq \nfrac{1}{\sqrt{2m}} \cdot \Deg 1$ and let $\cI$ be the identity over $\cS_D,$ i.e. $\cI \defeq I - \nfrac{1}{2m} \cdot \Deg 1 1^T \Deg.$
 For a positive $\e$ and a symmetric matrix $M \in \R^{V \times V},$ we define 
$$
U_{\e}(A) \defeq 2m \cdot \frac{\Degin e^{-(2m \cdot \e) \cdot \Degin M \Degin } \Degin}{\cI \bullet e^{-(2m \cdot \e) \Degin M \Degin}}.
$$
The following fact about $\cI$ will also be needed:
\begin{fact}\label{fct:identity}
$ \cI = 2m \cdot \Degin L(K_V) \Degin.$
\end{fact}

\subsection{\SDP Formulation} \label{sec:sdp}

We consider an $\SDP$ relaxation to the decision problem of determining whether the instance graph $G$ has a $b$-balanced cut of conductance at most $\gamma.$
The $\SDP$ feasibility program $\psdp(G,b,\gamma)$ appears in Figure $\ref{fig:sdp},$ where we also rewrite the program in matrix notation, using Fact \ref{fct:mean} and the definition of $R_i.$
\begin{figure}[htb]

\begin{minipage}[b]{0.5\linewidth}
\begin{alignat*}{4}
 \psdp(G, b, \gamma): & & \E_{\{i,j\} \in E}\;\norm{v_i-v_j}_2^2 \; \leq& \; 4\gamma   \\
 & & \E_{j \sim \mu}\; \norm{v_j- v_\avg}_2^2 \;=& \;1 \\
    &\forall i\in V \;   &  \norm{v_i- v_\avg}_2^2 \;\leq& \; \frac{1-b}{b}
\end{alignat*}
\end{minipage}
\begin{minipage}[b]{0.5\linewidth}
\begin{alignat*}{4}
 \psdp(G,b,\gamma):  & &  \frac{1}{m} \cdot L \bullet X \; \leq& \; 4\gamma \\
  & &  L(K_V) \bullet X \; =& \; 1 \\
     \forall i\in V & \quad  &  R_i \bullet X \; \leq& \; \frac{1-b}{b}
\end{alignat*}
\end{minipage}
\label{fig:sdp}
\caption{$\SDP$ for $b$-\BS}
\end{figure}
$\psdp$ can be seen as a scaled version of the balanced-cut $\SDP$ of \cite{ARV}, modified by replacing $v_\avg$ for the origin and removing the triangle-inequality constraints. The first change makes our $\psdp$ invariant under translation of the embeddings and makes the connection to spectral methods more explicit. Indeed, the first two constraints of \psdp now exactly correspond to the standard eigenvector problem, with the addition of the $R_i$ constraint ideally forcing all entries in the eigenvector not to be too far from the mean, just as it would be the case if the eigenvector exactly corresponded to a balanced cut. 
The removal of the triangle-inequality constraints causes $\psdp$ to only deal with the spectral structure of $L$ and not to have a flow component.
For the rest of the paper, denote by $\Delta$ the set $\{X \in \R^{V \times V}, X \succeq 0 : L(K_V) \bullet X = 1\}.$

\noindent
The following simple lemma establishes that $\psdp$ is indeed a relaxation for the integral decision question  and is proved in Section \ref{app:basic}. 
\begin{lemma}[SDP is a Relaxation]\label{lem:relax} If there exists a $b$-balanced cut $S$ with $\phi(S) \leq \gamma,$ then $\psdp(G,b,\gamma)$ has a feasible solution.
\end{lemma}

\noindent
\alg will use the primal-dual approach of \cite{AK} to determine the feasibility of $\psdp(G,b, \gamma).$ When \psdp is infeasible, \alg will output a solution to  the dual $\dsdp(G,b,\gamma),$ shown in Figure \ref{fig:dspd}. 

\begin{figure}[htb]
\begin{alignat*}{4}
\dsdp(G,b,\gamma):\quad &  \alpha - \frac{1-b}{b} \sum_{i \in V} \beta_i > 4\gamma\\
 &  \frac{1}{2m} \cdot L + \sum_{i \in V} \beta_i R_i-\alpha L(K_V) \succeq 0\\
 & \alpha \in \R, \; \; \beta \geq 0 
\end{alignat*} \caption{$\dsdp(G,b,\gamma)$ feasibility problem}
\label{fig:dspd}
\end{figure} 
\noindent
In the rest of the paper, we are going to use the following shorthands for the dual constraints
$$
V(\alpha, \beta) \defeq  \alpha - \frac{1-b}{b} \sum_{i \in V} \beta_i  , \;
M(\alpha, \beta) \defeq \frac{L}{2m}  + \sum_{i \in V} \beta_i R_i -\alpha L(K_V).
$$
Notice that $V(\alpha, \beta)$ is a scalar, while $M(\alpha, \beta)$ is a matrix in $\mathbb{R}^{V \times V}.$
Given $X \succeq 0,$ a choice of $(\alpha, \beta)$ such that $V(\alpha, \beta) > 4\gamma$ and $M(\alpha, \beta) \bullet X \geq 0$ corresponds to a hyperplane separating $X$ from the feasible region of $\psdp(G,b,\gamma)$ and constitutes a certificate that $X$ is not feasible. 

\noindent
Ideally, \alg would produce a feasible solution to \psdp and then round it to a balanced cut. However, as discussed in \cite{AK}, it often suffices to find a solution ``close'' to feasible for the rounding procedure to apply. In the case of \psdp, the concept of ``closeness'' is captured by the notion of {\it roundable} solution.
\begin{Definition}[Roundable Embedding]\label{def:roundable}
 Given an embedding $\{v_i\}_{i \in V},$ let $R = \{i \in V : \norm{v_i - v_\avg}_2^2 \leq 32 \cdot \nfrac{(1-b)}{b}\}.$ We say that $\{v_i\}_{i \in V}$ is a {\it roundable} solution to $\psdp(G, b, \gamma)$ if:
\begin{itemize}
\item $\E_{\{i,j\} \in E}\;\norm{v_i-v_j}_2^2 \; \leq \; 2\gamma $,
\item $\E_{j \sim \mu}\; \norm{v_j- v_\avg}_2^2 \;= \;1 $,
\item $  \E_{\{i,j\} \sim \mu_R \times \mu_R }\norm{v_i - v_j}_2^2 \geq \nfrac{1}{64}.$
\end{itemize} 
\end{Definition}

\noindent
A roundable embedding can be converted into a balanced cut of the conductance required by Theorem \ref{thm:main} by using a standard projection rounding, which is a simple extension of an argument already appearing in \cite{ARV} and \cite{AK}. The rounding procedure {\sc ProjRound} is described precisely in Section \ref{app:round}, where the following theorem is proved.
\begin{theorem}[Rounding Roundable Embeddings]\label{thm:stdround}
If $\{v_i \in \R^d \}_{i \in V}$ is a roundable solution to $\psdp(G,b,\gamma)$, then {\sc ProjRound}$(\{v_i\}_{i \in V}, b)$ produces a $\Omega_b(1)$- balanced cut of conductance $O_b\left(\sqrt{\gamma}\right)$ with high probability in time $\tilde{O}(n d + m).$
\end{theorem}

\subsection{Primal-Dual Framework} \label{sec:pd}

\begin{figure*}[htb]
  	\begin{tabularx}{\textwidth}{|X|}
    \hline
	 \vspace{1mm}

  {\bf \textsc{Input:}} An instance graph $G=(V,E),$ a balance value $b \in (0, \nfrac{1}{2}]$ such that $b = \Omega(1),$ a conductance value $\gamma \in (0,1).$
	\vspace{1mm}

	Let $\e = \nfrac{1}{130} .$ 
	For $t = 1, 2, \ldots, T = O\left(\frac{\log n}{\gamma}\right):
$
	\begin{itemize}
	\item Compute the embedding $\{\tilde{v}^{(t)}_i\}_{i \in V}$ corresponding to $\tilde{X}^{(t)} = \tilde{U}_\e\left(\sum_{j=1}^{t-1} P^{(j)}\right).$ If $t=1,$ $\tilde{X}^{(1)} = \tilde{U}_\e\left(0\right) = \nfrac{2m}{n-1} \cdot D^{-1}.$
   \item Execute {\sc Oracle}$\left(G, b, \gamma, \{\tilde{v}^{(t)}_i\}_{i \in V}\right).$
	
	\item If {\sc Oracle}\xspace finds that $\{\tilde{v}^{(t)}_i\}_{i \in V}$ is roundable, run {\sc ProjRound}$\left(G,b, \{\tilde{v}^{(t)}_i\}_{i \in V}\right),$ output the resulting cut and terminate. 
	\item Otherwise, {\sc Oracle}\xspace outputs coefficients $\left(\alpha^{(t)}, \beta^{(t)}\right)$ and cut $B^{(t)}.$ 
	
	\item Let $C^{(t)} \defeq \bigcup_{i=1}^t B^{(i)}.$ If $C^{(t)}$ is $\nfrac{b}{4}$-balanced, output $C^{(t)}$ and terminate.

	\item Otherwise, let $ P^{(t)} \defeq -\nfrac{\e}{6} \cdot \left(M\left(\alpha^{(t)}, \beta^{(t)}\right) + \gamma L(K_V)\right)$ and proceed to the next iteration.
  \end{itemize}

Output $S = \bigcup_{t=1}^T B^{(t)}.$ Also output  $\bar{\alpha} =\nfrac{1}{T} \sum_{t=1}^T  \alpha^{(t)}$ and $\bar{\beta} = \nfrac{1}{T} \sum_{t=1}^T \beta^{(t)}.$ 
\\
\\
\hline 
\end{tabularx}
  \caption{The \alg Algorithm}
  \label{fig:algorithm}
\end{figure*}

\paragraph{Separation Oracle.} The problem of checking the feasibility of a {\SDP} can be reduced to that of, given a candidate solution $X,$ to check whether it is close to feasible and, if not, provide a certificate of infeasibility in the form of a hyperplane separating $X$ from the feasible set. The algorithm performing this computation is known as a separation oracle. 
Arora and Kale show that the original feasiblity problem can be solved very efficiently if there exists a separation oracle obeying a number of conditions. We introduce the concept of {\it good} separation oracle to capture these conditions for the program $\psdp(G,\beta, \gamma).$

\begin{Definition}[Good Separation Oracle] An algorithm is a {\it good} separation oracle if, on input some representation of $X,$ the algorithm either finds $X$ to be a  roundable solution to $\psdp(G, b, \gamma)$ or outputs coefficents $\alpha, \beta$ such that $V(\alpha, \beta) \geq \nfrac{3}{4} \cdot \gamma,$  $M(\alpha, \beta) \bullet X \geq \nfrac{1}{64} \cdot \gamma$ and $ -\gamma L(K_V) \preceq  M(\alpha, \beta) \preceq 3 L(K_V).$
\end{Definition}

\paragraph{Algorithmic Scheme.} 
We adapt the techniques of \cite{AK} to our setting, where we require feasible solutions to be in $\Delta$ rather than having trace equal to 1.  The argument is a simple modification of the anaylsis of \cite{AK} and in \cite{Steurer}.
The algorithmic strategy of \cite{AK} is to produce a sequence of candidate primal solutions $X^{(1)}, \ldots, X^{(T)}$ iteratively, such that $X^{(t)}  \in \Delta$ for all $t.$ 

Our starting point $X^{(1)}$ will be the solution $\nfrac{2m}{n-1} \cdot D^{-1}.$ 
At every iteration, a {\it good} separation oracle {\sc Oracle}\xspace will take $X^{(t)}$ and either guarantee that $X^{(t)}$ is roundable or output coefficents $\alpha^{(t)}, \beta^{(t)}$ certifying the infeasiblity of $X^{(t)}.$ The algorithm makes use of the information contained in $\alpha^{(t)}, \beta^{(t)}$ by updating the next candidate solution as follows:
\begin{align} \label{eqn:update}
P^{(t)} \defeq -\nfrac{\e}{6} \cdot \left(M(\alpha^{(t)}, \beta^{(t)}) + \gamma L(K_V)\right) \\ 
 X^{(t+1)} \defeq U_\e\left( \sum_{i=1}^t P^{(i)} \right) \notag
\end{align}
where $\e > 0$ is a parameter of the algorithm. The following is immediate.
\begin{lemma}\label{lem:delta} For all $t > 0,$ $X^{(t)} \in \Delta.$
\end{lemma}

\noindent
Following \cite{AK}, we prove that, after a small number of iterations this algorithm either yields a roundable embedding or a feasible solution to \dsdp$(G,b, \Omega(\gamma)).$
We present the proof in Section \ref{app:ak} for completeness.
\begin{theorem}[Iterations of Oracle, \cite{AK}]\label{thm:ak}
Let $\e = \nfrac{1}{130}.$ Assume that the procedure $\textsc{Oracle}$ is a {\it good} separation oracle . Then, after $T=O\left(\nfrac{\log n}{\gamma}\right)$ iterations of the update of Equation \ref{eqn:update},  we either find a roundable solution to $\psdp(G,b,\gamma)$ or the coefficents  $\bar{\alpha} =\nfrac{1}{T} \sum_{t=1}^T  \alpha^{(t)}$ and $\bar{\beta} = \nfrac{1}{T} \sum_{t=1}^T \beta^{(t)}$ are a feasible solution to $\dsdp(G,b, \nfrac{3}{16} \cdot\gamma).$ 
\end{theorem}

\paragraph{Approximate Computation.} 
Notice that, while we are seeking to construct a nearly-linear-time algorithm, we cannot hope to compute $X^{(t)}$ exactly and explicitly, as just maintaining the full $X^{(t)}$ matrix requires quadratic time in $n.$  
Instead, we settle for a approximation $\tilde{X}^{(t+1)}$  to $X^{(t+1)}$ which we define as
$$
\tilde{X}^{(t+1)} = \tilde{U}_\e\left( \sum_{i=1}^t P^{(i)} \right).
$$
The function $\tilde{U}_\e$ is a randomized approximation to $U_\e$ obtained by applying the Johnson-Linderstrauss dimension reduction to the embedding corresponding to $U_\e.$ $\tilde{U}_\e$ is described in full in Section \ref{app:exp}, where we also prove the following lemma about the accuracy and sparsity of the approximation. It is essentially the same argument appearing in \cite{Kthesis} applied to our context.

\begin{lemma}[Approximate Computation]\label{lem:approx}
Let $\e = \nfrac{1}{130}.$
For a matrix $M \in \R^{V \times V},$ $M \succeq 0,$ let $\tilde{X} \defeq \tilde{U}_\e(M)$ and $X \defeq U_\e(M).$
\begin{enumerate}
\item $\tilde{X} \succeq 0$ and $\tilde{X} \in \Delta.$
\item The embedding $\{\tilde{v}_i\}_{i \in V}$ corresponding to $\tilde{X}$ can be represented in $d=O(\log n)$ dimensions.
\item $\{\tilde{v}_i \in \R^d \}_{i \in V}$ can be computed in time $\tilde{O}(t_M + n).$
\item for any graph $H=(V,E_H)$, with high probability
$$ (1 - \nfrac{1}{64}) \cdot L(H) \bullet X -\tau \leq L(H) \bullet \tilde{X} \leq (1 + \nfrac{1}{64}) \cdot L(H) \bullet X + \tau,$$ and, for any vertex $i \in V,$
$$ (1 - \nfrac{1}{64}) \cdot R_i \bullet X -\tau \leq R_i \bullet \tilde{X} \leq (1 + \nfrac{1}{64}) \cdot R_i \bullet X + \tau,
$$

\end{enumerate}
where $\tau \leq O(\nfrac{1}{\poly(n)}).$
\end{lemma}
This lemma shows that  $\tilde{X}^{(t)}$ is a close approximation to $ X^{(t)}.$ 
We will use this lemma to show that {\sc Oracle}\xspace can receive  $\tilde{X}^{(t)}$ as input, rather than  $X^{(t)},$ and still meet the conditions of Theorem \ref{thm:ak}.
In the rest of the paper, we assume that $\tilde{X}^{(t)}$ is represented by its corresponding embedding $\{\tilde{v}^{(t)}_i\}_{i \in V}.$

\paragraph{The Oracle.} {\sc Oracle}\xspace is described in Figure \ref{fig:oracle}.
 We show that {\sc Oracle}\xspace on input $\tilde{X}^{(t)}$ meets the condition of Theorem \ref{thm:ak}. Moreover, we show that {\sc Oracle}\xspace obeys an additional condition, which, combined with the dual guarantee of Theorem \ref{thm:ak} will yield the correlation property of {\sc BalCut}.
\begin{theorem}[Main Theorem on {\sc Oracle}] \label{thm:oracle}
On input $\tilde{X}^{(t)},$ {\sc Oracle} runs in time $\tilde{O}(m)$ and is a {\it good} separation oracle for $X^{(t)}.$ Moreover, the cut $B$ in Step \ref{stp:cut} is guaranteed to exist.

\end{theorem}

\begin{figure*}[htb]
  	\begin{tabularx}{\textwidth}{|X|}
    \hline
  	\begin{enumerate}
\item {\bf \textsc{Input:}} The embedding $\{\tilde{v}_i\}_{i \in V},$ corresponding to $\tilde{X} \in \Delta.$ Let $r_i = \norm{\tilde{v}_i - \tilde{v}_\avg}_2$ for all $i \in V.$
Denote $R \defeq \{i \in V : r_i^2 \leq 32 \cdot \nfrac{(1-b)}{b}\}.$
\item {\sc Case 1}:  $\E_{\{i,j\} \in E}\;\norm{\tilde{v}_i-\tilde{v}_j}_2^2  \geq  2\gamma.$ 
	Output $\alpha = \gamma,$ $\beta = 0$ and $B = \emptyset.$
	\item {\sc Case 2}: not {\sc Case 1} and $\E_{\{i,j\} \sim \mu_R \times \mu_R }\norm{v_i - v_j}_2^2 \geq \delta.$ Then $\{\tilde{v}_{i}\}_{i\in V}$ is roundable, as $\tilde{X} \in \Delta$ implies $\E_{j \sim \mu}\; r_j^2 = 1.$
\item \label{stp:cut} {\sc Case 3}: not {\sc Case 1} or {\sc 2}. Relabel the vertices of $V$ such that $r_1 \geq r_2 \geq \ldots \geq r_n$ and let $S_i=\{1, \ldots,i\}$ be the $j$-th sweep cut of $r.$ 
Let $z$ the smallest index such that $\mu(S_z) \geq \nfrac{b}{8}.$ 	Let $B$ the most balanced sweep cut among $\{S_1, \ldots, S_{z-1}\}$ such that $\phi(B) \leq 2048 \cdot \sqrt{\gamma}.$  Output $\alpha = \nfrac{7}{8} \gamma,$ $\beta_i = \mu_i \cdot \gamma$ for $i \in B$ and $\beta_i = 0$ for $i \notin B.$ Also output the cut $B.$ 
 \end{enumerate}\\
   \hline
    \end{tabularx}
  \caption{{\sc Oracle}}
  \label{fig:oracle}
\end{figure*}

\paragraph{Proof of Main Theorem.}  We are now ready to prove Theorem \ref{thm:main}. To show the overlap condition,  we consider the dual condition implied by Theorem \ref{thm:ak} together with the cut $B^{(t)}$ and the values of the coefficents output by the {\sc Oracle}.

\begin{proof}[Proof of Theorem \ref{thm:main}]
If at any iteration $t,$ the embedding $\{\tilde{v}_i^{(t)}\}_{i \in V}$ corresponding to $\tilde{X}^{(t)}$  is roundable, the standard projection rounding {\sc ProjRound}\xspace produces a cut of balance $\Omega_b(1)$ and conductance $O_b(\sqrt{\gamma})$ by Theorem \ref{thm:stdround}. Similarly, if for any $t,$ $C^{(t)}$ is $\nfrac{b}{4}$-balanced, \alg satisfies the balance condition  in Theorem \ref{thm:main}, as $\phi(C^{(t)}) \leq O(\sqrt{\gamma})$ because $C^{(t)}$ is the union of cuts of conductance at most $O(\sqrt{\gamma}).$

Otherwise, after $T=O\left(\nfrac{\log n}{\gamma}\right)$ iterations, by Theorem \ref{thm:ak}, we have that  $\bar{\alpha} =\nfrac{1}{T} \sum_{t=1}^T  \alpha^{(t)}$ and $\bar{\beta} = \nfrac{1}{T} \sum_{t=1}^T \beta^{(t)}$ constitute a feasible solution  $\dsdp(G,b, \nfrac{3}{16} \cdot\gamma).$ This implies that
$ 
M(\bar{\alpha}, \bar{\beta}) \succeq 0,$ i.e.
\begin{equation}\label{eqn:dual}
\frac{1}{m} \cdot L +  \sum_{i \in V} \bar{\beta}_i R_i -  \bar{\alpha} L(K_V) \succeq 0 .
\end{equation}

For any cut $C$ such that $\mu(C) \leq \nfrac{1}{2}$ and $\phi(C) \leq \nfrac{\gamma}{16},$ let the embedding $\{ u_i \in \R \}_{i \in V}$ be defined as $u_i = \sqrt{\nfrac{\mu(\bar{C})}{\mu(C)}}$ for $i \in C$ and $u_i = -\sqrt{\nfrac{\mu(C)}{\mu(\bar{C})}}$ for $i \notin C.$
Then $u_\avg = 0$ and $\E_{i \sim \mu} \norm{u_i - u_\avg}_2^2 = 1.$ Moreover, $\E_{\{i,j\} \in E} \norm{u_i - u_j}_2^2 = \nfrac{1}{m} \cdot \nfrac{\card{E(C, \bar{C})}}{\mu(C)\mu(\bar{C})} \leq 4 \cdot \phi(C) \leq \nfrac{\gamma}{4}.$  
Let $U$ be the Gram matrix of  $\{ u_i \in \R \}_{i \in V}.$

We apply the lower bound of Equation \ref{eqn:dual} to $U.$ By Facts \ref{fct:mean} and \ref{fct:star}.
\begin{align*}
\E_{\{i,j\} \in E} \norm{u_i - u_j}_2^2 + \sum_{i \in V} \bar{\beta}_i \norm{u_i - u_\avg}^2_2 - \bar{\alpha} \E_{i \sim \mu}\norm{u_i - u_\avg}_2^2 \\
= M(\bar{\alpha}, \bar{\beta}) \bullet U \geq 0
\end{align*}

Recall that, by the definition of {\sc Oracle}, for all $t \in [T],$ $ \alpha^{(t)} \geq \nfrac{7}{8} \cdot \gamma$ and $\beta^{(t)}_i = \mu_i \cdot \gamma$ for $i \in B^{(t)}$ and $\beta^{(t)}_i = 0$ for $i \notin B^{(t)}.$ Hence,
\begin{eqnarray*}
\nfrac{\gamma}{4} + \nfrac{\gamma}{T} \cdot \sum_{t=1}^T  \left( \mu(B^{(t)} \cap C) \cdot \nfrac{\mu(\bar{C})}{\mu(C)} + \mu(B^{(t)} \cap \bar{C}) \cdot \nfrac{\mu(C)}{\mu(\bar{C})} \right) \\ - \nfrac{7}{8} \cdot \gamma \geq 0 
\end{eqnarray*}

Dividing by $\gamma$ and using the fact that $\mu(C) \leq \nfrac{1}{2}$ and $\mu(\bar{C}) \leq 1,$ we obtain
$$
\nfrac{1}{T} \cdot \sum_{t=1}^T  \left(\frac{\mu(B^{(t)} \cap C)}{\mu(C)} + \frac{\mu(B^{(t)} \cap \bar{C})}{2 \cdot \mu(\bar{C})} \right) \geq \left(\nfrac{7}{8} - \nfrac{1}{4}\right)=\nfrac{5}{8}. 
$$
Now,
$$
\frac{\mu(S \cap C)}{\mu(C)} + \frac{\mu(S\cap \bar{C})}{2 \cdot \mu(\bar{C})} \geq \nfrac{1}{T} \cdot \sum_{t=1}^T  \left(\frac{\mu(B^{(t)} \cap C)}{\mu(C)} + \frac{\mu(B^{(t)} \cap \bar{C})}{2 \cdot \mu(\bar{C})} \right), 
$$
so that we have
$$
\frac{\mu(S \cap C)}{\mu(C)} + \frac{\mu(S\cap \bar{C})}{2 \cdot \mu(\bar{C})} \geq \nfrac{5}{8}.
$$
Moreover, being the union of cuts of conductance $O(\sqrt{\gamma}),$ $S$ also has $\phi(S) \leq O(\sqrt{\gamma}).$
As $\mu(S) \leq \nfrac{b}{4},$ $\nfrac{\mu(S\cap \bar{C})}{2 \cdot \mu(\bar{C})} \leq \mu(S) \leq \frac{b}{4} \leq \nfrac{1}{8}.$ This finally implies that
$$
\frac{\mu(S \cap C)}{\mu(C)} \geq \nfrac{1}{2}.
$$
Finally,  both {\sc ProjRound} and {\sc Oracle} run in time $\tilde{O}(m)$ as the embedding is $O(\log n)$ dimensional. By Lemma \ref{lem:approx}, the update at time $t$ can be performed in time $t_M,$ where $M=\sum_{i=1}^{t-1} P^{(i)}.$ This is a matrix of the form $a \cdot L + \sum_{i \in V} b_i R_i + c L(K_V).$ The first two terms can be multiplied by a $O(\log n)$ vectors in time $\tilde{O}(m),$ while the third term can be decomposed as $L(K_V) = \E_{i \sim \mu} R_i$ by Fact \ref{fct:mean} and can therefore be also multiplied in time $\tilde{O}(m).$
Hence, each iteration runs in time $\tilde{O}(m),$ which shows that the total running time is $\tilde{O}(\nfrac{m}{\gamma})$ as required.

\end{proof}

\section{Proof of Main Theorem on {\sc Oracle}} \label{sec:mp}

\subsection{Preliminaries}

The following is a variant of the sweep cut argument of Cheeger's Inequality \cite{FAN}, tailored to ensure that a constant fraction of the variance of the embedding is contained inside the output cut.
For a vector $x \in \R^{V},$ let ${\rm supp}(x)$ be the set of vertices where $x$ is not zero. 
\begin{lemma} \label{lem:cheeger}
Let $x \in \R^V, x \geq 0,$ such that  $x^T L x \leq \lambda$ and $\mu(\supp(x)) \leq \nfrac{1}{2}.$ Relabel the vertices so that $x_1 \geq x_2 \geq \ldots \geq x_{z-1} > 0$ and $x_{z} = \ldots = x_n = 0.$ For $i \in [z-1],$ denote by $S_i \subseteq V,$ the sweep cut $\{1, 2, \ldots, i\}.$ 
Further, assume that $\sum_{i=1}^n d_i x_i^2 \leq 1,$ and, for some fixed $k \in [z-1],$ $\sum_{i=k}^{n} d_i x_i^2 \geq \sigma.$
Then, there is a sweep cut $S_h$ of $x$ such that $ z-1 \geq h \geq k$ and $\phi(S_h) \leq \nfrac{1}{\sigma} \cdot \sqrt{ 2 \lambda}.$
\end{lemma}

\noindent
We will also need the following simple fact.
\begin{fact} \label{fct:triangle}
Given $v,u, t \in \R^d,$ $\left(\norm{v-t}_2 - \norm{u-t}_2\right)^2 \leq \norm{v - u}_2^2.$
\end{fact}

\subsection{Proof of Theorem \ref{thm:oracle}}

\begin{proof}
Notice that, by Markov's Inequality,  $\mu(\bar{R}) \leq \nfrac{b}{(32 \cdot(1-b))} \leq \nfrac{b}{16}.$ Recall that $\tau = O\left(\nfrac{1}{\poly(n)}\right).$
\begin{itemize}
\item {\sc Case 1}: $\E_{\{i,j\} \in E}\;\norm{\tilde{v}_i-\tilde{v}_j}_2^2 \;= \frac{1}{2m} \cdot L \bullet \tilde{X} \; \geq \; 2\gamma .$ We have $V(\alpha, \beta) \geq \gamma$ and, by Lemma \ref{lem:approx}, 
$$
M(\alpha, \beta) \bullet X \geq (1 - \nfrac{1}{64}) \cdot  2\gamma - \gamma - \tau \geq \nfrac{1}{64} \cdot \gamma.
$$
\item {\sc Case 2}: $\E_{\{i,j\} \sim \mu_R \times \mu_R }\norm{v_i - v_j}_2^2 \geq \nfrac{1}{64}.$ Then $\{\tilde{v}_{i}\}_{i\in V}$ is {\it roundable} by Definition \ref{def:roundable}.

\item {\sc Case 3}: $\E_{\{i,j\} \sim \mu_R \times \mu_R }\norm{v_i - v_j}_2^2 < \nfrac{1}{64}.$  This means that, by Fact \ref{fct:subset}, $L(K_R) \bullet \tilde{X} < \nfrac{1}{2} \cdot \mu(R)^2 \cdot \nfrac{1}{64} < \nfrac{1}{128}.$ Hence, by Fact \ref{fct:star},
\begin{align*}
\sum_{i \in \bar{R}} \mu_i R_i \bullet \tilde{X} =  
\sum_{i \in \bar{R}} \mu_i r_i
\geq \mu(R) - \nfrac{1}{128} \geq 1 - \nfrac{1}{32} - \nfrac{1}{128}\\ \geq 1 - \nfrac{5}{128}.
\end{align*}

We then have $\bar{R} = S_g$ for some $g \in [n],$ where we also  denote by  $S_{g}$  the $g$ largest coordinates dictated by the sweep cut $S_{g}.$ 
Let $k \leq z$ be the the vertex in $R$ such that $\sum_{j=1}^{k} \mu_j r_j \geq  (1 - \nfrac{1}{128}) \cdot (1 - \nfrac{5}{128})$ and  $\sum_{j=k}^{g}\mu_j r_j \geq \nfrac{1}{128} \cdot (1 - \nfrac{5}{128}).$ 
By the definition of $z,$ we have $k \leq g < z$ and $r_z^2 \leq \nfrac{8} {b} \leq 16 \cdot \nfrac{(1-b)}{b}.$ Hence, we have $r_z \leq \nfrac{1}{2} \cdot r_i,$ for all $i \geq g.$
Define the vector $x$ as $x_i \defeq \nfrac{1}{2m} \cdot (r_i - r_z)$ for $i \in S_z$ and $r_i \defeq 0$ for $i \notin S_z.$ 
Notice that:
\begin{align*}
x^T L x = \sum_{\{i,j\} \in E} (x_i - x_j)^2 
\leq \nfrac{1}{2m} \cdot  \sum_{\{i,j\} \in E} (r_i - r_j)^2 \\
\stackrel{\rm Fact \; \ref{fct:triangle}}{\leq} 
\nfrac{1}{2m} \cdot \sum_{\{i,j\} \in E} \norm{\tilde{v}_i - \tilde{v}_j}_2^2 \leq  2\gamma.
\end{align*}
Also, $x \geq 0$ and $\mu(\supp(x)) \leq \nfrac{b}{8} \leq \nfrac{1}{2},$ by the definition of $z.$ Moreover,
$$
\sum_{i=1}^n d_i x_i^2 = \nfrac{1}{2m} \cdot \sum_{i=1}^z d_i (r_i - r_z)^2 \leq \nfrac{1}{2m}\cdot \sum_{i=1}^z d_i r_i^2 = 1, 
$$ and
\begin{align*}
\sum_{i=k}^n d_i x_i^2 = \nfrac{1}{2m}\cdot  \sum_{i=k}^z d_i (r_i -r_z)^2 \\
\geq \nfrac{1}{2m} \cdot \sum_{i=k}^g d_i (r_i - \nfrac{1}{2} \cdot r_i)^2 
\\= \nfrac{1}{2m} \cdot \nfrac{1}{4} \cdot \sum_{i=k}^g   d_i r_i^2  
\\ 
\geq \nfrac{1}{512} \cdot (1 - \nfrac{5}{128}) \geq \nfrac{1}{1024}
\end{align*}
Hence, by Lemma \ref{lem:cheeger}, there exists a sweep cut $S_h$ with $ z > h \geq k,$ such that $\phi(S_h) \leq 2048 \cdot \sqrt{\gamma} .$ This shows that $B,$ as defined in Figure \ref{fig:oracle} exists. Moreover, it must be the case that $S_h \subseteq B.$ As $h \geq k,$ we have
$$
\sum_{i \in B} \mu_i r_i^2 \geq \sum_{i \in S_h} \mu_i r_i^2 \geq \sum_{i=1}^k \geq (1 - \nfrac{1}{128}) \cdot (1-\nfrac{5}{128}) \geq 1 - \nfrac{3}{64}.
$$
Recall also that, by the construction of $z,$ $\mu(B) \leq \nfrac{b}{8}.$
Hence, we have 
$$
V(\alpha, \beta) = \nfrac{7}{8} \cdot  \gamma - \nfrac{(1-b)}{b} \cdot \mu(B) \cdot \gamma \geq (\nfrac{7}{8} - \nfrac{1}{8}) \cdot \gamma \geq \nfrac{3}{4} \gamma.
$$
$$
M(\alpha, \beta) \bullet X \geq (1 - \nfrac{1}{64}) \cdot  (1 - \nfrac{3}{64}) \gamma  - \nfrac{7}{8} \gamma - \tau \geq \nfrac{1}{64} \cdot \gamma.
$$
\end{itemize}
This completes all the three cases.
Notice that in every case we have:
$$
\nfrac{1}{2m} \cdot L - \gamma L(K_V) \preceq M(\alpha, \beta) \preceq \nfrac{1}{2m} \cdot L + \gamma L(K_V) .
$$
\noindent
Hence,
$$
 -\gamma L(K_V) \preceq M(\alpha, \beta) \preceq 3 L(K_V).
$$

\noindent
Finally,  using the fact that $\{\tilde{v_i}\}_{i \in V}$ is embedded in $O(\log n)$ dimensions, we can compute $L \bullet \tilde{X}$ in time $\tilde{O}(m).$ $L(K_R) \bullet \tilde{X}$ can also be computed in time $\tilde{O}(n)$ by using the decomposition  $\E_{\{i,j\} \sim \mu_R \times \mu_R }\norm{v_i - v_j}_2^2 = 2 \cdot\E_{i \sim \mu_R} \norm{v_i - v_{\avg_R}}^2_2,$ where $v_{\avg_R}$ is the mean of vectors representing vertices in $R.$ 
The sweep cut over $r$ takes time $\tilde{O}(m).$ Hence, the total running time is $\tilde{O}(m).$ 
\end{proof}

\addcontentsline{toc}{section}{References}
\bibliographystyle{plain}
\bibliography{balanced}

\appendix
\section{Appendix}

\subsection{Proof of Corollary \ref{cor:cut}}\label{app:cut}
\begin{lemma}
If only the second condition in Theorem \ref{thm:st} holds, then $G$ has no $\Omega_b(1)$-balanced cut of conductance $O(\gamma).$
\end{lemma}
\begin{proof}
We may assume that the cut $S$ output by \alg is not $\Omega_b(1)$-balanced. Then, by the second condition, any cut $T$ with $\Vol{T} \leq \nfrac{1}{2} \cdot \vol(G)$ and $\phi(T) \leq O(\gamma)$ must have $\nfrac{\Vol{S \cap T}}{\Vol{T}} \geq \nfrac{1}{2}.$
Hence, $\Vol{T} \leq 2 \cdot \Vol{S \cap T} \leq 2 \cdot \Vol{S} \leq \Omega_b(1).$ This implies that there are no $\Omega_b(1)$-balanced cuts of conductance less than $O(\gamma).$
\end{proof}

\subsection{Proof of Basic Lemmata} \label{app:basic}

\begin{proof}[Proof of Lemma \ref{lem:relax}]
For a $b$-balanced cut $(S, \bar{S})$ with $\phi(S) \leq \gamma.$ Without loss of generality, assume $\mu(S) \leq \nfrac{1}{2}.$ Consider the one-dimensional solution assigning $v_i = \sqrt{\nfrac{\mu(\bar{S}}{\mu(S)}}$ to $i \in S$ and $v_i = -\sqrt{\nfrac{\mu(S}{\mu(\bar{S})}}$ to $i \in \bar{S}.$ Notice that $v_\avg = 0$ and that $\norm{v_i - v_j}_2^2 = \nfrac{1}{\mu(S) \mu(\bar{S})}$ for $i \in S, j \notin S.$
We then have:
\begin{itemize}
 \item \begin{align*}
	  \E_{\{i,j\} \in E}\;\norm{v_i-v_j}_2^2 = \frac{1}{m} \cdot \frac{|E(S, \bar{S})|}{\mu(S) \mu(\bar{S})}\leq 2 \cdot \frac{|E(S, \bar{S})|}{2m \cdot \mu(S) \mu(\bar{S})} \\
\leq 4 \cdot \phi(S) \leq 4 \gamma.
\end{align*}
\item
$$ 
 \E_{i \sim \mu}\norm{v_i - v_\avg}_2^2 
= \mu(S) \cdot \nfrac{\mu(\bar{S})}{\mu(S)} + \mu(\bar{S}) \cdot \nfrac{\mu(S)}{\mu(\bar{S})}  = 1. 
$$ 
\item for all $i \in V,$ 
$$
 \norm{v_i- v_\avg}_2^2 \leq \frac{\mu(\bar{S})}{\mu(S)} \leq \frac{1-b}{b},
$$

\noindent
where the last inequality follows as $S$ is $b$-balanced.
\end{itemize}
\end{proof}

\begin{proof}[Proof of Lemma \ref{lem:cheeger}]
For all $i \in V,$ let $y_i = x_i^2.$
By Cauchy-Schwarz,
\begin{align*}
\sum_{\{i,j\} \in E} |y_i - y_j| = \sum_{\{i,j\} \in E,\; x_i \geq x_j} (x_i - x_j) (x_i + x_j) \\
\leq \sqrt{\left(\sum_{\{i,j\} \in E} (x_i - x_j)^2 \right) \cdot \left(\sum_{\{i,j\} \in E} (x_i + x_j)^2 \right)}.
\end{align*}
Hence,
$$
\sum_{\{i,j\} \in E} |y_i - y_j| \leq \sqrt{\lambda \cdot 2\left(\sum_{i \in V} d_i x_i^2. \right)} \leq \sqrt{2\lambda}.
$$
Then, let $\phi$ be the conductance of the least conductance cut among $S_{k+1}, S_{k+2}, \ldots, S_h.$
\begin{align*}
\sum_{\{i,j\} \in E} |y_i - y_j| \\
\geq \sum_{i= k+1}^h \card{E(S_i, \bar{S}_i)} \cdot (y_i - y_{i+1}) \geq \phi  \sum_{i= k+1}^h \Vol{S_i} \cdot (y_i - y_{i+1})\\
 \geq \phi \sum_{i= k+1}^h d_i y_i = \sum_{i=k+1}^h d_i x_i^2 = \phi \sum_{i \in V - S_k} d_i x_i^2 = \phi \sigma.
\end{align*}
Hence, $ \phi \leq \nfrac{1}{\sigma} \cdot \sqrt{2\lambda}.$

\end{proof}

\subsection{Primal-Dual Framework}\label{app:ak}

\subsubsection{Preliminaries}
Recall that $\cS_D$ is the subspace of $\R^V$ orthogonal to $\hat{v} = \nfrac{1}{\sqrt{2m}} \cdot \Deg 1$ and that $\cI$ be the identity over $\cS_D.$

Define 
$$
E_\e(Y) \defeq \frac{\exp(-\e Y)}{\cI \bullet \exp(-\e Y)}.
$$

The following simple facts will be useful.
\begin{fact}\label{fct:exp}
Let $M$ be a symmetric matrix in $\R^{V \times V}$ such that $M \hat{v} = 0.$ Then, $\exp(M) \hat{v} = \hat{v}.$
\end{fact}

\begin{theorem}\label{thm:expo}
Let $\e > 0$ and let $Y^{(1)}, \ldots, Y^{(T)}$ be a sequence of symmetric matrices in $\R^{V \times V}$ such that, for all $i \in [T],$ $ Y_i \hat{v} = 0$ and $0 \preceq Y_i \preceq I.$
Then:
$$
\sum_{y=1}^T Y^{(t)} \succeq \left((1-\e)\sum_{t=1}^T Y^{(t)} \bullet Z^{(t)} - \frac{\log n}{\e}\right) \cI,
$$
where $Z^{t+1} = E_\e(\sum_{i=1}^t Y^{(i)}).$
\end{theorem}

\subsubsection{Proofs}

 In the following, let $M^{(t)} \defeq 2m \cdot \Degin P^{(t)} \Degin.$ Also, let
$
W^{(t)} = E_\e\left(\sum_{i=1}^{t-1} M^{(i)}\right).
$
Then,  we have 
\begin{align*}
X^{(t)} = U_\e\left(\sum_{i=1}^{t-1} P^{(i)}\right) = 2m \cdot \Degin E_\e\left(\sum_{i=1}^{t-1} M^{(i)}\right) \Degin\\ = 2m \cdot \Degin W^{(t)} \Degin.
\end{align*}

\begin{proof}[Proof of Lemma \ref{lem:delta}]
By Fact \ref{fct:identity},
$$
 L(K_V) \bullet X^{(t)} = \left(\Deg \cI \Deg\right) \bullet \left(\Degin W^{(t)} \Degin\right) = \cI \bullet W^{(t)} = 1.
$$
\end{proof}
              
We are now ready to complete the proof of Theorem \ref{thm:ak}.                                                                                                                                                                                                                                                                                                                                                                                               \begin{proof}[Proof of Theorem \ref{thm:ak}]
Suppose that {\sc Oracle}\xspace outputs coefficents $\left(\alpha^{(t)}, \beta^{(t)}\right)$ for $T$ iterations.
Then for all $t \in [T],$ by the definition of {\it good} oracle $W^{(t)} \bullet M^{(t)}$
\begin{align*}
 = \left(\nfrac{1}{2m} \cdot \Deg X^{(t)} \Deg)\right) \bullet \left(2m \cdot \Degin P^{(t)} \Degin\right) = X^{(t)} \bullet P^{(t)} \\
\geq \frac{1}{6} \cdot \left( M\left(\alpha^{(t)}, \beta^{(t)}\right) \bullet X^{(t)} + \gamma L(K_V) \bullet X^{(t)} \right)\\
\geq \nfrac{1}{6} \cdot \left(\nfrac{1}{64} \cdot \gamma + \gamma\right) = \nfrac{1}{6} \cdot \nfrac{65}{64} \cdot \gamma. 
\end{align*}
Notice that, as $\left(\alpha^{(t)}, \beta^{(t)}\right)$ are the output of a {\it good} oracle, we have
$0  \preceq P^{(t)} \preceq L(K_V),$ which implies $0 \preceq M^{(t)} \preceq \cI \preceq I.$
Hence, we can apply Theorem \ref{thm:expo} to obtain:
$$
\sum_{i=1}^T M^{(t)} \succeq \left((1-\e)\sum_{t=1}^T M^{(t)} \bullet W^{(t)} - \frac{\log n}{\e}\right)\cI.
$$
This implies
$$
\sum_{i=1}^T P^{(t)} \succeq \left((1-\e) \cdot \nfrac{T}{6} \cdot \nfrac{65}{64} \cdot \gamma - \frac{\log n}{\e}\right)L(K_V),
$$
and by the definition of $P^{(t)},$
\begin{align*}
\frac{1}{6}  \cdot \sum_{i=1}^T \left( M\left(\alpha^{(t)}, \beta^{(t)}\right) + \gamma L(K_V) \right) \\
\succeq \left((1-\e) \cdot \nfrac{T}{6} \cdot \nfrac{65}{64} \cdot \gamma - \frac{\log n}{\e}\right)L(K_V).
\end{align*}
Hence,
$$
M(\bar{\alpha}, \bar{\beta}) \succeq \nfrac{1}{T} \cdot \left((1-\e) \cdot T \cdot  \nfrac{65}{64} \cdot \gamma - \frac{6 \log n}{\e} - \gamma\cdot T \right)L(K_V)
$$
and
$$
M(\bar{\alpha}, \bar{\beta}) \succeq   \left((1-\e) \cdot  \nfrac{65}{64} \cdot \gamma - \frac{6 \log n}{\e T} - \gamma \right)L(K_V).
$$
By picking $\e \defeq \nfrac{1}{130}$ and $T \defeq 6 \cdot 129 \cdot \nfrac{1}{\e} \cdot \nfrac{\log n}{\gamma} = O\left( \nfrac{\log n}{\gamma}\right),$ we obtain $M(\bar{\alpha}, \bar{\beta}) \succeq 0 \cdot L(K_V).$
As $M(\bar{\alpha}, \bar{\beta}) 1 = 0,$ this also implies
$$
M(\bar{\alpha}, \bar{\beta}) \succeq 0.
$$
Finally, by the definition of {\it good} oracle, $V(\bar{\alpha}, \bar{\beta}) \geq \nfrac{3}{4} \cdot \gamma.$ Hence, $(\bar{\alpha}, \bar{\beta})$ is a solution to $\dsdp(G,b, \nfrac{3}{16}\cdot\gamma).$
\end{proof}

\subsection{Projection Rounding}\label{app:round}

\begin{figure*}[h]
  \begin{tabularx}{\textwidth}{|X|}
    \hline
  \begin{enumerate}
  \item {\bf \textsc{Input:}} An embedding $\{v_i \in \R^d \}_{i \in V} ,$ $b\in (0,\nfrac{1}{2}].$
  \item Let $c = \Omega_b(1)$ be a constant to be fixed in the proof. 
\item For $t=1,2, \ldots, O(\log n)$:
\begin{enumerate}
\item Pick a unit vector $u$ uniformly at random from $\mathbb{S}^{d-1}$ and let $x
\in \R^n$ with $x_i \defeq \sqrt{d} \cdot  { u^T v_i}$.
\item Sort the vector $x.$ Assume w.l.og. that $x_1 \geq x_2 \geq \ldots \geq x_n$. Define $S_i \defeq \{j\in [n]:x_j \geq x_i\}$.
   
  \item  Let $S^{(t)} \defeq (S_i,\bar{S_i})$ which minimizes
  $ \phi(S_i)$ among sweep-cuts for which  $\vol(S_i) \in [c\cdot 2m, (1-c)\cdot 2m].$
\end{enumerate}
\item {\bf \textsc{Output:}} The cut $S^{(t)}$ of least conductance over all choices of $t.$ 
  \end{enumerate}\\
   \hline
    \end{tabularx}

  \caption{{\sc ProjRound}}
  \label{fig:rounding}
\end{figure*}

The description of the rounding algorithm {\sc ProjRound} is given in Figure \ref{fig:rounding}. We remark that during the execution of \alg the embedding $\{ v_i \in \R^d\}_{i \in V}$ will be represented by a projection over $d=O\left(\log n \right)$ random directions, so that it will suffice to take a balanced sweep cut of each coordinate vector.

We now present  the proof of Theorem \ref{thm:stdround}. 
The constants in this argument were not optimized to preserve the simplicity of the proof.

\subsubsection{Preliminaries.} We will make use of the following simple facts. Recall that for $y \in \R,$ $\sgn(y)=1$ if $y \geq 0,$ and $-1$ otherwise.

\begin{fact}\label{fct:ab}
For all $y,z \in \R,$ $(y+z)^2 \leq 2(y^2+z^2).$ \end{fact}
\begin{proof}
$2(y^2+z^2)-(y+z)^2=(y-z)^2 \geq 0.$ 
\end{proof}

\begin{fact}\label{fct:abs}
For all $y \geq z \in \R,$ $\left|\sgn(y)\cdot y^2-\sgn(z)\cdot z^2 \right| \leq  (y-z)(|y|+|z|).$ \end{fact}
\begin{proof} ${\rm } \\$

\begin{enumerate}
\item If $\sgn(y)=\sgn(z),$ then  $|\sgn(y)\cdot y^2-\sgn(z)\cdot z^2 | = |y^2-z^2|=(y-z)\cdot |y+z|=(y-z)(|y|+|z|)$ as $y \geq z.$
\item If $\sgn(y) \neq \sgn(y),$ then since $y \geq z,$ $(y-z)=|y|+|z|.$ 
Hence, $|\sgn(y)\cdot y^2-\sgn(z)\cdot z^2 | = y^2+z^2\leq (|y|+|z|)^2=(y-z)(|y|+|z|).$
\end{enumerate}
\end{proof}

\begin{fact}\label{fct:ab2}
For all $y\geq z \in \R,$ $(y-z)^2 \leq 2(\sgn(y)\cdot y^2-\sgn(z)\cdot z^2).$ \end{fact}
\begin{proof} ${\rm } \\$

\begin{enumerate}
\item If $\sgn(y)=\sgn(z),$  $(y-z)^2=y^2+z^2-2yz \leq y^2+z^2-2z^2=y^2-z^2$ as $y \geq z.$ Since $\sgn(y)=\sgn(z),$  $y^2-z^2 \leq 2(\sgn(y)\cdot y^2-\sgn(z)\cdot z^2).$
\item If $\sgn(y) \neq \sgn(z),$ $(y-z)^2=(|y|+|z|)^2 \leq 2(|y|^2+|z|^2) = 2(\sgn(y)\cdot y^2-\sgn(z)\cdot z^2).$ Here, we have used Fact \ref{fct:ab}. 
\end{enumerate}
\end{proof}

We also need the following standard facts.
\begin{fact}\label{fct:prob}
Let $v \in \R^d$ be a vector of length $\ell$ and $u$ a unit vector chosen uniformly at random in $\mathbb{S}^{t-1}$. Then,
\begin{enumerate}
\item $\E_u \; \left(v^T u\right)^2 = \frac{\ell^2}{d},$ and 
\item  for $0 \leq \delta \leq 1$,
$\Pr_u \left[ \sqrt{d} \cdot \abs{v^T u} \leq \delta \ell \right] \leq 3\delta.$
\end{enumerate}
\end{fact}


\begin{fact}\label{fct:antimarkov}
Let $Y$ be a non-negative random variable such that $\Pr [Y \leq K] =1$  and $\E[Y] \geq \delta.$ Then,   
$$\Pr[Y \geq \delta/2] \geq \frac{\delta }{2K }.$$  
\end{fact}

The following lemma about projections will be crucial in the proof of Theorem \ref{thm:stdround}. It is a simple adaptation of an argument appearing in \cite{ARV}.
\begin{lemma}[Projection]\label{lem:projection}
Given a roundable embedding $\{v_i \in \R^d\}_{i\in V},$  consider the embedding $x \in \R^n$ such that $x_i \defeq  \sqrt{d} \cdot  {u^T v_i},$ where $u \in \mathbb{S}^{d-1},$ and assume without loss of generality that $x_1 \geq \ldots \geq  x_n.$    Then, there exists  $c \in (0,b]$ such that with  probability $\Omega_b(1)$ over the choice of $u \in \mathbb{S}^{d-1}$, the following conditions hold simultaneously:
\begin{enumerate}
\item $\E_{\{i,j\} \in E} (x_i - x_j)^2 \leq O_b\left(\E_{\{i,j\} \in E} \norm{v_i - v_j}^2 \right) = O_b (\gamma)$,
\item $\E_{i \sim \mu} (x_i - x_\avg)^2 = O_b(1),$ and 
\item there exists $1\leq l\leq n$ with $\vol(\{1,\ldots,l\}) \geq c\cdot \vol(G)$ and, there exists $l \leq r \leq n$ such that $\vol(\{r, \ldots, n\}) \geq c\cdot  \vol(G)$ such that $x_l - x_r \geq \Omega_b(1)$.
\end{enumerate}
\end{lemma}

\begin{proof}
We are going to lower bound the probability, over $u$,  of each of (1), (2) and (3) in the lemma and then apply the union bound. 

\paragraph{Part (1).} 
By applying Fact \ref{fct:prob} to $v = v_i - v_j$ and noticing $\sqrt{d} \cdot \abs{v^T u} = \abs{x_i - x_j}$ , we have
$$ \E_u \E_{\{i,j\} \in E} (x_i - x_j)^2 = \E_{\{i,j\} \in E} \norm{v_i - v_j}^2.$$
Hence, by Markov's Inequality, for some $p_1$ to be fixed later
$$ \Pr_u\left[\E_{\{i,j\} \in E} (x_i - x_j)^2 \geq \nfrac{1}{p_1} \cdot \E_{\{i,j\} \in E} \norm{v_i - v_j}^2\right] \leq p_1. $$

\paragraph{Part (2).}
$$ \E_{u} \E_{i \sim \mu} (x_i - x_\avg)^2 \stackrel{{\rm Fact \; \ref{fct:prob}-(1)}}{=}  \E_{u} \E_{i \sim \mu} \norm{v_i - v_\avg}^2 \stackrel{{\rm roundability}}{=} 1.$$ Hence, for some $p_2$ be fixed later 
$$  \Pr_u\left[ \E_{i \sim \mu} (x_i - x_\avg)^2 \geq \nfrac{1}{p_2} \cdot  \E_{i \sim \mu} \norm{v_i - v_\avg}^2 \right] \leq p_2.$$
 
\paragraph{Part (3).}
Let $R \defeq \{i \in V : \norm{v_i -v_\avg}^2 \leq 32 \cdot \nfrac{(1-b)}{b}\}.$ 
Let $\sigma \defeq 4 \cdot \sqrt{2} \sqrt{\nfrac{(1-b)}{b}}.$
By Markov's Inequality, $\mu(\bar{R}) \leq \nfrac{1}{\sigma^2}.$
As $\{v_i\}_{i \in V}$ is roundable, for all $i,j \in R,$ $ \norm{v_i - v_j} \leq 2 \sigma.$ Hence,   $\norm{v_i - v_j} \geq \nfrac{1}{2 \sigma} \cdot  \norm{v_i - v_j}^2$ for such $i,j \in R.$ 
This, together with the roundability of $\{v_i\}_{i \in V}$, implies that
$$ \E_{\{i,j\} \sim \mu_R \times \mu_R} \norm{v_i - v_j}   \geq  \nfrac{1}{128\sigma} .$$
For any $k \in R$, we can apply the triangle inequality for the Euclidean norm as follows
\begin{align*}
 \E_{\{i,j\} \sim \mu_R \times \mu_R}   \norm{v_i - v_j}    \leq  \E_{\{i,j\} \sim \mu_R \times \mu_R}    \left(\norm{v_i - v_k} + \norm{v_k - v_j}\right) \\
\leq 2 \cdot  \E_{i \sim \mu_R}   \norm{v_i - v_k}.
\end{align*}
Hence, for all $k \in R$
$$\E_{i \sim \mu_R}     \norm{v_i - v_k} \geq  \nfrac{1}{256\sigma}.$$
Let $R_k$ be the set $\left\{i \in R : \norm{v_i - v_k} \geq \nfrac{1}{512\sigma}  \right\}$. 
Since $\norm{v_i - v_k} \leq 2\sigma$, applying Fact \ref{fct:antimarkov} yields that, for all $k \in R$, 
$$\Pr_{i \sim \mu_R }\left[i \in R_k\right] \geq \nfrac{1}{1024\sigma^2}. $$
For all vertices $i \in R_k$, by Fact \ref{fct:prob}
$$\Pr_u\left[\abs{x_k - x_i}\geq \nfrac{1}{9} \cdot  \nfrac{1}{512\sigma} = \nfrac{1}{4608\sigma} \right]  \geq \frac{2}{3}.$$
Let $\delta \defeq \nfrac{1}{2}\cdot \nfrac{1}{4608\sigma}  = \nfrac{1}{9216\sigma}.$ Consider the event $\cE \defeq \{  i \in R_k \wedge \abs{x_i - x_k} \geq 2 \cdot \delta \}.$
Then,
\begin{align*}
\Pr_{u, \;\{i,k\} \sim \mu_R \times \mu_R} [\cE] =  \Pr_{\{i,k\} \sim \mu_R \times \mu_R}[i \in R_k ] \cdot \Pr_u[\abs{x_i - x_k} \geq 2 \cdot \delta \;|\; i \in R_k] 
\\
\geq \frac{1}{1024\sigma^2} \cdot \frac{2}{3} = \frac{1}{1536\sigma^2} \defeq \rho.
\end{align*}

Hence, from Fact \ref{fct:antimarkov}, with probability at least $\nfrac{\rho}{2}$ over directions $u,$ for a fraction $\nfrac{\rho}{2}$ of pairs $\{i,k\} \in R \times R,$  $\abs{x_k - x_i} \geq 2\cdot \delta.$ 
Let $\nu$ be the median value of $\{x_i\}_{i \in V}$.
Let $L \defeq \left\{i : x_i \leq \nu - \delta \right\} $ and $H \defeq \left\{i : x_i \geq \nu + \delta \right\}$. Any pair $\{i,j\} \in R \times R$ with $\abs{x_i - x_j} \geq 2\cdot \delta$ has at least one vertex in $L \cup H$. 
Hence,
$$\mu( L \cup H) \geq \nfrac{1}{2} \cdot \nfrac{\rho}{2} \cdot \mu(R)^2 \geq \nfrac{\rho}{4} \cdot\left(\nfrac{\sigma^2 - 1}{\sigma^2}\right)^2 \geq \nfrac{\rho}{16} = \Omega_b(1) .$$
Assume $\mu(L) \geq \nfrac{\rho}{32} ,$ otherwise, apply the same argument to $H$.
Let $l$ be the largest index in $L.$ 
For all $i \in L$ and $j$ such that $x_j \geq \nu$, we have $\abs{x_i - x_j} \geq \delta$. (Similarly, let $r$ be the smallest index in $H$.)
This implies that, 
$$\abs{x_{l} - x_{\floor{\nfrac{n}{2}}}} \geq \delta$$
with probability at least $\nfrac{\rho}{2} = \Omega_b(1)$, satisfying the required condition.
Let $p_3$ be the probability that this event does not take place. Then,
$$p_3 \leq 1 - \nfrac{\rho}{2} .$$

\noindent
To conclude the proof, notice that the probability that all three conditions do not hold simultaneously is, by a union bound, at most $p_1 + p_2 + p_3$. Setting $p_1 = p_2 = \nfrac{\rho}{5} = \Omega_b(1)$, we satisfy the first and third conditions and obtain
$$ p_1 + p_2 + p_3 \leq 1 - \rho\cdot \left(\nfrac{1}{2} - \nfrac{1}{5} - \nfrac{1}{5}\right)  \leq 1 - \nfrac{\rho}{10}.$$
Hence, all conditions are satisfied at the same time with probability at least $ \nfrac{\rho}{10} = \Omega_b(1)$.
\end{proof}
\noindent
From this proof, it is possible to see that the parameter $c$ in our rounding scheme should be set to $\nfrac{\rho}{32}.$

We are now ready to give a proof of Theorem \ref{thm:stdround}. It is essentially a variation of the proof of Cheeger's Inequality, tailored to produce balanced cuts.


\begin{proof}[Proof of Theorem \ref{thm:stdround}]
For this proof, assume that $x$ has been translated so that $x_\avg = 0.$ Notice that the guarantees of $\ref{lem:projection}$ still apply.
Let $x, l, r$ and $c$ be as promised by Lemma \ref{lem:projection}. 
For $z \in \R,$ let $\sgn(z)$ be $1$ if $z \geq 0$ and $-1$ otherwise.
Let 
$$ y_i \defeq  \sgn (x_i) \cdot  x_i^2.$$
 Hence,
\begin{align*}
\E_{\{i,j\} \in E} |y_i - y_j| & \stackrel{\rm Fact \; \ref{fct:abs}}{\leq}  \E_{\{i,j\} \in E} (|x_i - x_j|)\cdot (|x_i|+|x_j|) \\
& \leq \sqrt{ \E_{\{i,j\} \in E} (x_i - x_j)^2\cdot \E_{\{i,j\} \in E} (|x_i|+|x_j|)^2} \\
 & \stackrel{\rm Fact \;  \ref{fct:ab}}{\leq}   \sqrt{2\cdot  \E_{\{i,j\} \in E} (x_i - x_j)^2\cdot \E_{\{i,j\} \in E} (x_i^2+x_j^2)} \\
& =  \sqrt{2\cdot  \E_{\{i,j\} \in E} (x_i - x_j)^2\cdot \frac{2m}{m} \cdot \E_{i \sim \mu}{x_i}^2} \\
%
%
& =  \sqrt{4\cdot  \E_{\{i,j\} \in E} (x_i - x_j)^2\cdot \E_{i \sim \mu}{x_i}^2}  \\
&  \stackrel{{\rm Lemma \; \ref{lem:projection}-(1),(2)}}{\leq} O_b\left(\sqrt{\gamma}\right).  
\end{align*}
Now we lower bound $\E_{\{i,j\} \in E} |y_i - y_j|.$ Notice that if $x_i \geq x_j,$ then $y_i \geq y_j$ and vice-versa. Hence, 
$$y_1 \geq \ldots \geq y_n.$$
Let $S_i \defeq \{1,\ldots,i\}$  and let $\phi$ be the minimum conductance of $S_i$ over all $l \leq i \leq .r$   
\begin{align*}
\E_{\{i,j\} \in E} |y_i - y_j| &=  \frac{1}{|E|} \sum_{i=1}^{n-1} |E(S_i,\bar{S}_i)| \cdot (y_{i}-y_{i+1}) \\
& \geq  \phi  \cdot   \sum_{ l  \leq i \leq r}  \frac{\min \{ \vol(S_i), \vol (\bar{S}_i)\}}{|E|}(y_i-y_{i+1}) \\
& \stackrel{{\rm Lemma \; \ref{lem:projection}-(3)}}{=}  \Omega_b(1) \cdot \phi \cdot  \sum_{ l \leq i \leq r} (y_i-y_{i+1}) \\
& {\geq}  \Omega_b(1) \cdot   \phi \cdot (y_l - y_r) \\
& \stackrel{{\rm Fact \; \ref{fct:ab2}}}{\geq}   \Omega_b(1) \cdot \phi \cdot  (x_l - x_r)^2 \\
& \stackrel{{\rm Lemma \; \ref{lem:projection}}}{\geq}   \Omega_b(1) \cdot  \phi .
\end{align*}
Hence, $\phi \leq O_b(\sqrt{\gamma})$ with constant probability over the choice of projection vectors $u.$ Repeating the projection $O(\log n)$ times and picking the best balanced cut found yields a high probability statement. 
Finally, as the embedding is in $d$ dimensions, it takes $\tilde{O}(nd)$ time to compute the projection. After that, the one-dimensional embedding can be sorted in time $\tilde{O}(n)$ and the conductance of the relevant sweep cuts can be computed in time $O(m),$ so that the total running time is $\tilde{O}(nd + m).$
\end{proof}

\subsection{Proof of Lemma \ref{lem:approx}}\label{app:exp}

\subsubsection{Preliminaries}
For the rest of this section the norm notation will mean the norm in the subspace $\cS_D.$ Hence $\norm{A}=\norm{\cI A \cI}.$
We will need the following lemmata.

\begin{lemma}[Johnson-Lindenstrauss]\label{lem:jl}
Given an embedding $\{v_i \in \R^n\}_{i \in V}$, $V=[n],$ let $u_1, u_2, \ldots, u_k$, be vectors sampled independently uniformly from the $n-1$-dimensional sphere of radius $\sqrt{\nfrac{n}{k}}.$ Let $U$ be the $k \times t$ matrix having the vector $u_i$ as $i$-th row and let $\tilde{v}_i \defeq  U v_i$. Then, for $k_\delta \defeq O(\nfrac{\log n}{\delta^2}),$ for all $i, j \in V$ 
$$
(1 - \delta) \cdot \norm{v_i - v_j}^2 \leq \norm{\tilde{v}_i - \tilde{v}_j}^2 \leq (1 + \delta)\cdot  \norm{v_i - v_j}^2
$$
and
$$
(1 - \delta) \cdot \norm{v_i}^2 \leq \norm{\tilde{v}_i}^2 \leq (1 + \delta) \cdot\norm{v_i}^2.
$$
\end{lemma}

\begin{lemma}[\cite{Kthesis}]\label{lem:expv} 
There exists an algorithm {\sf EXPV} which, on input of a matrix $A \in \R^{n \times n}$, a vector $u \in \R^n$ and a parameter $\eta$, computes a vector $v \in \R^n$, such that $\norm{v - e^{-A}u} \leq \normt{e^{-A}} \cdot \eta$ in time $O(t_A \log^3(\nfrac{1}{\eta})).$
\end{lemma}
\noindent
The algorithm {\sf EXPV} is described in \cite{Kthesis} and \cite{YPS} .

\subsubsection{Proof}

We define the $\tilde{U}_\e$ algorithm in Figure \ref{fig:approx} and proceed to prove Lemma \ref{lem:approx}.
\begin{figure*}[h]
  \begin{tabularx}{\textwidth}{|X|}
    \hline
  \begin{itemize}
  \item {\bf \textsc{Input:}} A matrix $M \in \R^{n \times n}.$
  \item Let $\eta \defeq O(\nfrac{1}{\poly(n)}).$ Let $\delta \defeq \nfrac{1}{512}$ and $\e=\nfrac{1}{130}.$
	\item For $k_\delta$ as in Lemma \ref{lem:jl}, sample $k_\delta$ vectors $u_1, \ldots, u_{k_\delta} \in \R^n$ as in Lemma \ref{lem:jl}.
	\item Let $A \defeq (2m \cdot \e) \cdot \Degin M \Degin.$
 \item For $1 \leq i \leq k_\delta,$ compute vectors $b_i \in \R^n, b_i \defeq {\sf EXPV}(\nfrac{1}{2} \cdot A, \Degin u_i, \eta).$

 \item Let $B$ be the matrix having $b_i$ as $i$-th row,  and let $\tilde{v}_i$ be the $i$-th column of $B$. Compute $Z \defeq \E_{\{i,j\} \in \mu \times \mu} \norm{\tilde{v}_i -\tilde{v}_j}^2 = L(K_V) \bullet B^T B.$
	\item  Return $\tilde{X} \defeq \nfrac{1}{Z} \cdot B^T B$, by giving its correspoding embedding, i.e., $\{\nfrac{1}{\sqrt{Z}} \cdot \tilde{v}_i\}_{i \in V}.$

  \end{itemize}\\
   \hline
    \end{tabularx}

  \caption{The $\tilde{E}_\e$ algorithm}
  \label{fig:approx}
\end{figure*}

\begin{proof}
We verify that the conditions required hold.
\begin{itemize}
\item By construction, $\tilde{X} \succeq 0$, as $\tilde{X} = \nfrac{1}{Z} \cdot B^T B,$ and $L(K_V) \bullet \tilde{X} = 1.$
\item $\tilde{X}=\left(\nfrac{1}{\sqrt{Z}} \cdot B\right)^T\left(\nfrac{1}{\sqrt{Z}} \cdot B
\right)$ and $B$ is a $k_\delta \times n$ matrix, with $k_\delta = O(\log n),$ by Lemma \ref{lem:jl}.
\item We perform $k_{\delta} = O(\log n)$ calls to the algorithm {\sf EXPV}, each of which takes time $\tilde{O}(t_A)=\tilde{O}(t_M + n).$ Sampling the vectors $\{u_i\}_{1,\ldots, k_\delta}$ and computing $Z$ also requires $\tilde{O}(n)$ time. Hence, the total running time is $\tilde{O}(t_M +n).$

\item Let $U$ be the $k_\delta \times n$ matrix having the sampled vectors $u_1, \ldots, u_{k_\delta}$ as rows. 
Let $\{v_i\}_{i \in V}$
be the embedding corresponding to matrix $Y \defeq \Degin e^{-A} \Degin,$ i.e., $v_i$ is the $i$-th column of $Y^{\nfrac{1}{2}}.$ Notice that $X= \nfrac{Y}{L(K_V) \bullet Y}.$
Define $\hat{v}_i \defeq Uv_i$ for all $i$ and let $\hat{Y}$ be the Gram matrix corresponding to this embedding, i.e., $\hat{Y} \defeq (Y^{\nfrac{1}{2}})^TU^TU(Y^{\nfrac{1}{2}}).$
Also, let $\tilde{Y}$ be the Gram matrix corresponding to the embedding $\{\tilde{v}_i\}_{i\in V},$ i.e., $\tilde{Y} = B^T B$ and $\tilde{X} = \nfrac{\tilde{Y}}{L(K_V) \bullet \tilde{Y}}.$
We will relate $Y$ to $\hat{Y}$ and $\hat{Y}$ to $\tilde{Y}$ to complete the proof.

First, by Lemma \ref{lem:jl}, applied to $\{v_i\}_{i \in V}$, with high probability, for all $H$
$$
(1-\delta) \cdot L(H) \bullet Y \leq L(H) \bullet \hat{Y} \leq (1+\delta) \cdot L(H) \bullet Y
$$
and for all $i \in V$
$$
(1-\delta) \cdot  R_i \bullet Y \leq R_i \bullet \hat{Y} \leq (1+\delta) \cdot R_i \bullet  Y.
$$
In particular, this implies that $ (1-\delta) \cdot \cI \bullet Y \leq \cI \bullet \hat{Y} \leq (1 + \delta) \cdot \cI \bullet Y.$
Hence,
$$
\frac{1-\delta}{1+\delta} \cdot L(H) \bullet X \leq L(H) \bullet \hat{X} \leq \frac{1+\delta}{1-\delta} \cdot L(H) \bullet X
$$
and for all $i$
$$
\frac{1-\delta}{1+\delta} \cdot  R_i \bullet X \leq R_i \bullet \hat{X} \leq \frac{1+\delta}{1-\delta} \cdot R_i \bullet X.
$$

Now we relate $\hat{Y}$ and $\tilde{Y}$. Let $E \defeq \left(\tilde{Y}^{\nfrac{1}{2}} - \hat{Y}^{\nfrac{1}{2}}\right) \Deg.$ 
By Lemma \ref{lem:expv} 
\begin{align*}
\normt{E}^2 \leq \norm{E}^2_F = \sum_{i} \norm{d_i \tilde{v}_i - \hat{v}_i}^2  \leq 2m \cdot \normt{e^{\nfrac{\e}{2}\cdot  A}}^2 \cdot  \eta^2 \\
\leq 2m \cdot \normt{Y^{\nfrac{1}{2}}\Deg}^2 \cdot \eta^2 \leq  (2m)^2 \cdot L(K_V) \bullet Y \cdot \eta^2.
\end{align*}
This also implies
\begin{align*}
\normt{E} \cdot \normt{\hat{Y}^{\nfrac{1}{2}}\Deg} \leq \norm{E}_F \cdot \norm{\hat{Y}^{\nfrac{1}{2}}\Deg}_F 
\\
\leq \left(\sqrt{2m} \cdot \normt{Y^{\nfrac{1}{2}}\Deg} \cdot \eta \right) \cdot \sqrt{\sum_{i} d_i \norm{\hat{v}_i -\hat{v}_\avg}^2 } 
\\
\leq 2m \cdot \eta \cdot (1+\delta) \cdot L(K_V) \bullet Y.
\end{align*}
As $\Deg\left(\tilde{Y} - \hat{Y}\right)\Deg = E^T E + (\hat{Y}^{\nfrac{1}{2}})^T E + E^T \hat{Y}^{\nfrac{1}{2}},$ we have
\begin{align*}
\normt{\Deg\left(\tilde{Y} - \hat{Y}\right)\Deg} 
\\
\leq \normt{E^T E + (\hat{Y}^{\nfrac{1}{2}})^T E + E^T \hat{Y}^{\nfrac{1}{2}}} 
\\
\leq \sqrt{3} \left(\normt{E}^2 + 2 \cdot \normt{E} \normt{\hat{Y}^{\nfrac{1}{2}}}\right) \\
\leq 
9 \cdot (2m)^2 \cdot (1+\delta) \cdot L(K_V) \bullet Y \cdot \eta.
\end{align*}
and
\begin{align*}
\abs{L(K_V) \bullet (\tilde{Y} - \hat{Y})} \\
\leq  L(K_V) \bullet (E^T E) + 2 \cdot \abs{L(K_V) \bullet (E^T \hat{Y}^{\nfrac{1}{2}})}
\\
 \leq \nfrac{1}{2m} \cdot \norm{E}^2_F + \nfrac{2}{2m} \cdot \norm{E}_F \norm{\hat{Y}^{\nfrac{1}{2}}}_F 
\\
\leq 3 \cdot 2m \cdot (1+\delta) \cdot L(K_V) \bullet Y \cdot \eta.
\end{align*}

Finally, combining these bounds we have
\begin{align*}
\Norm{\tilde{X} - \hat{X}}_2 =  \Norm{\frac{\tilde{Y}}{L(K_V) \bullet \tilde{Y}} - 
\frac{\hat{Y}}{L(K_V) \bullet \hat{Y}}}_2 
\\
\leq 
\Norm{\frac{\tilde{Y}}{L(K_V) \bullet \tilde{Y}} - \frac{\tilde{Y}}{L(K_V) \bullet \hat{Y}}}_2
\\ + 
\Norm{\frac{\tilde{Y}}{L(K_V) \bullet \hat{Y}} - \frac{\hat{Y}}{L(K_V) \bullet \hat{Y}}}_2\\
  \leq  
 \frac{\normt{\tilde{Y}} \cdot \abs{L(K_V) \bullet \tilde{Y} - L(K_V) \bullet \hat{Y}}}{L(K_V) \bullet \tilde{Y} \cdot L(K_V) \bullet \hat{Y}} +
\frac{\normt{\tilde{Y} - \hat{Y}}}{L(K_V) \bullet \hat{Y}} 
\\
\leq \frac{2m \cdot \abs{L(K_V) \bullet \tilde{Y} - L(K_V) \bullet \hat{Y}} + \normt{\tilde{Y} - \hat{Y}}}{L(K_V) \bullet \hat{Y}} \\
 \leq \frac{12 \cdot (2m)^2 \cdot (1+\delta) \cdot L(K_V) \bullet Y \cdot \eta}{(1-\delta) \cdot L(K_V) \bullet Y}
\\ \leq 12 \cdot (2m)^2 \cdot \nfrac{1+\delta}{1-\delta} \cdot \eta 
\\ \leq  O(\nfrac{1}{\poly(n)})
\end{align*}
by taking $\eta$ sufficiently small in $O(\nfrac{1}{\poly(n)}).$

Hence, as $\normt{L(H)} \leq O(m)$ and $\normt{R_i} \leq O(m)$
$$
|L(H) \cdot \hat{X} - L(H) \cdot \tilde{X}| \leq  O(\nfrac{1}{\poly(n)})
$$ 
and
$$
|R_i \bullet \hat{X} - R_i \bullet\tilde{X} | \leq  O(\nfrac{1}{\poly(n)}).
$$
This, together with the fact that $\nfrac{1-\delta}{1+\delta} \geq 1 - \nfrac{1}{64}$ and $\nfrac{1+\delta}{1-\delta} \leq 1 + \nfrac{1}{64}$ completes the proof.
\end{itemize}

\end{proof}

\end{document}